\documentclass{lmcs}
\pdfoutput=1
\usepackage[utf8]{inputenc}

\usepackage{lastpage}
\lmcsdoi{21}{4}{10}
\lmcsheading{}{\pageref{LastPage}}{}{}%
{May~20,~2024}{Oct.~16,~2025}{}

\usepackage{color}
\usepackage{amsfonts,amssymb,amsmath,amsthm}
\usepackage{mathtools}
\usepackage{mathrsfs}
\usepackage{xparse}
\usepackage{thmtools}
\usepackage{bussproofs}

\usepackage{tikz-cd}
\usepackage{enumitem}
\usepackage{esvect}  %
\usepackage{wrapfig}

\usepackage[hidelinks]{hyperref}

\newenvironment{prf}{\begin{proof}}{\end{proof}}

\newenvironment{axioms}{\begin{enumerate}[labelsep=8pt,leftmargin=*,itemindent=2em,labelindent=1.0\parindent]}{\end{enumerate}}

\newcommand{\df}[1]{\emph{#1}}

\newcommand{\tjcol}{blue}
\newcommand{\dmcol}{purple}
\newcommand{\nscol}{orange}

\newenvironment{tjblock}{\par\medskip\color{\tjcol}}{\medskip}
\newenvironment{dmblock}{\par\medskip\color{\dmcol}}{\medskip}

\newenvironment{needswork}{\par\medskip\color{\nscol}}{\medskip}

\makeatletter
\def\namedlabel#1#2{\begingroup%
    #2%
    \def\@currentlabel{#2}%
    \phantomsection\label{#1}\endgroup
}
\def\labelhere#1#2{\begingroup%
    \def\@currentlabel{#2}%
    \phantomsection\label{#1}\endgroup
}
\makeatother

\newcommand{\Fraisse}{Fra\"{i}\-ss\'{e}}
\newcommand{\ef}{Ehren\-feucht--\Fraisse}
\newcommand\bimorph{heteromorphism}
\newcommand\Bimorph{Heteromorphism}

\newcommand{\ol}{\overline}
\newcommand{\wt}{\widetilde}

\newcommand\ee[1]{\enspace #1 \enspace}
\newcommand\mm[1]{\, #1 \,}
\newcommand{\tq}[1]{\mbox{#1}\quad}
\newcommand\qtq[1]{\quad\mbox{#1}\quad}

\newcommand\ete[1]{\ee{\mbox{#1}}}

\newcommand{\sue}{\subseteq}

\newcommand\obj{\mathrm{obj}}
\newcommand\id{\mathrm{id}}

\newcommand{\embed}{\rightarrowtail}
    \newcommand{\emb}{\embed}
\newcommand{\quot}{\twoheadrightarrow}

\renewcommand\vec[1]{\vv{#1}}
\newcommand\veci[1]{\vec{#1_i}\mkern1mu}
\newcommand\opveci[1]{\op(\vec{#1_i})}
\newcommand\opvec[1]{\op(\vec{#1})}
\newcommand\emcm{\mathbin{\rotatebox[origin=c]{270}{$\embed$}}}

\newcommand\Lo{\mathcal L}

\newcommand\parrow[1]{\Rrightarrow_{\exists^+ #1}}
\newcommand\earrow[1]{\Rrightarrow_{\exists #1}}
\newcommand\larrow[1]{\Rrightarrow_{#1}}
\newcommand\pequiv[1]{\equiv_{\exists^+ #1}}
\newcommand\cequiv[1]{\equiv_{\# #1}}
\newcommand\sequiv[1]{\equiv_{#1}}
\newcommand\lequiv[1]{\equiv_{#1}}

\newcommand{\fgST}[1]{#1\vert_{\tau}}

\newcommand{\merge}[1]{\stackrel{#1}{\vee}}

\newcommand{\lift}[1]{\widehat{#1}}          %

\newcommand\op{H} %
\newcommand\lop{\lift{\op}}      %

\newcommand{\univ}{\mathrm{u}}

\newcommand\sg{\sigma}
\newcommand\tr{\mathfrak t}

\newcommand\I{I}
\newcommand\trI{\tr^\I}

\newcommand\Con{\texttt{com}}
\newcommand\trCon{\tr^\Con}
\newcommand\sgCon{\sg^\Con}

\newcommand\counit{\varepsilon}
\newcommand{\comonad}[1]{\mathbb{#1}} %
\newcommand\C{\mathbb C} %
\newcommand\D{\mathbb D}

\newcommand{\Ek}{{\comonad{E}_{k}}}
\newcommand{\Pk}{{\comonad{P}_{k}}}
\newcommand{\Mk}{{\comonad{M}_{k}}}

\newcommand{\Cos}{\mathtt{Cos}}

\newcommand{\EM}[1]{\mathsf{EM}(#1)}
\newcommand{\EMF}[1]{F^{#1}}
\newcommand{\EMU}[1]{U^{#1}}

\newcommand{\Klei}[1]{\mathsf{Kl}(#1)}
    \def\Kl{\Klei}

\newcommand{\kirc}{\mathbin{\bullet}}
\newcommand{\klto}[1]{\to_{#1}}

\newcommand{\slice}{\mathbin{\downarrow}}
\newcommand{\cat}[1]{\mathcal{#1}}
\newcommand{\CA}{\cat{A}}
\newcommand{\CB}{\cat{B}}
\newcommand{\CC}{\cat{C}}
\newcommand{\CD}{\cat{D}}

\newcommand{\CS}{\cat{S}}
\newcommand{\CT}{\cat{T}}
\newcommand{\CU}{\cat{U}}

\newcommand{\Rel}{\mathcal R(\sg)}
\newcommand{\Rels}{\mathcal R_*(\sg)}
\newcommand{\R}{\mathcal R}
\newcommand{\Rs}{\mathcal R_*}
\newcommand{\Id}{{I\mkern-1mu d}}

\newcommand{\Paths}{\mathcal{P}}
\newcommand{\Pa}{\Paths}
\newcommand{\Qu}{\mathcal{E}}
\newcommand{\Em}{\mathcal{M}}

\newcommand{\struct}[1]{#1}
\newcommand{\As}{\struct{A}}
\newcommand{\Bs}{\struct{B}}
\newcommand{\Cs}{\struct{C}}
\newcommand{\Ps}{\struct{P}}
\newcommand{\Qs}{\struct{Q}}

\newcommand{\structn}[2]{#1_{#2}}
\newcommand{\pstructn}[3]{(\structn{#1}{#3}, #2_{#3})}
\newcommand{\Asan}[1]{\pstructn{A}{a}{#1}}
\newcommand{\Bsbn}[1]{\pstructn{B}{b}{#1}}

\newcommand{\Ac}{\alpha}
\newcommand{\Bc}{\beta}

\newcommand{\Pc}{\pi}
\newcommand{\Qc}{\rho}

\begin{document}

\title[A categorical account of composition methods in logic]{A categorical account of\texorpdfstring{\\}{} composition methods in logic}


\titlecomment{{\lsuper*}This is an extended version of the conference paper~\cite{JaklMS23} }
\thanks{Tom\'a\v s Jakl was funded by the EXPRO project 20-31529X awarded by the Czech Science Foundation and was also supported by the Institute of Mathematics, Czech Academy of Sciences (RVO 67985840), Dan Marsden was supported by EPSRC grant EP/T00696X/1: Resources and co-resources: a junction between categorical semantics and descriptive complexity, Nihil Shah was funded by UKRI under UK government's Horizon Europe funding guarantee grant EP/X028259/1.}

\author[T.~Jakl]{Tom\'a\v{s} Jakl \lmcsorcid{0000-0003-1930-4904}}[a]
\author[D.~Marsden]{Dan Marsden \lmcsorcid{0000-0003-0579-0323}}[b]
\author[N.~Shah]{Nihil Shah \lmcsorcid{0000-0003-2844-0828}}[c]
\address{Faculty of Information Technology, Czech Technical University} %
\email{tomas.jakl@cvut.cz}
\address{University of Nottingham} %
\email{dan.marsden@nottingham.ac.uk}
\address{University of Cambridge} %
\email{nas54@cam.ac.uk}

\begin{abstract}
    \noindent
    We present a categorical theory of the composition methods in finite model theory -- a key technique enabling modular reasoning about complex structures by building them out of simpler components.  The crucial results required by the composition methods are Feferman--Vaught--Mostowski (FVM) type theorems, which characterize how logical equivalence behaves under composition and transformation of models.

    Our results are developed by extending the recently introduced game comonad semantics for model comparison games. This level of abstraction allows us to give conditions yielding FVM type results in a uniform way. Our theorems are parametric in the classes of models, logics and operations involved. Furthermore, they naturally account for the existential and positive existential fragments, and extensions with counting quantifiers of these logics. We also reveal surprising connections between FVM type theorems, and classical concepts in the theory of monads.

    We illustrate our methods by recovering many classical theorems of practical interest, including a refinement of a previous result by Dawar, Severini, and Zapata concerning the 3-variable counting logic and cospectrality. To highlight the importance of our techniques being parametric in the logic of interest, we prove a family of FVM theorems for products of structures, uniformly in the logic in question, which cannot be done using specific game arguments.

    This is an extended version of the LiCS 2023 conference paper of the same name.
\end{abstract}

\maketitle

\section{Introduction}
\emph{Composition methods} constitute a family of techniques %
in finite model theory for understanding the logical properties of complex structures~\cite{libkin2004elements}. One works in a modular fashion, building a structure out of simpler components. The larger structure can then be understood in terms of the logical properties of its parts, and how they behave under the operations used in the construction.

The first result of this type was proved by Mostowski~\cite{mostowski1952direct}, who showed that the first-order theory of the product of two structures $A \times B$ is determined by the first-order theories of $A$ and $B$. Later, Feferman and Vaught famously proved a very general result for first-order logic, which included showing that the same holds true for infinite products and infinite disjoint unions of structures of the same type~\cite{feferman1959first}. Since then, many more Feferman--Vaught--Mostowski (FVM) theorems\footnote{The terser Feferman--Vaught type theorem is more common in the literature.} have been shown for various operations, logics and types of structures. These theorems have important applications in the theory of algorithms~\cite{makowsky2004algorithmic}, for example in the development of algorithmic meta-theorems such as Courcelle's theorem~\cite{courcelle2012graph}.

For our purposes, the typical form of an FVM theorem for a fixed $n$-ary operation $H$ and logics $L_1,\ldots,L_{n + 1}$ is as follows.
Given structures $A_1,\dots, A_n$ and $B_1, \dots, B_n$,
\begin{equation}
\label{eq:fvm-general}
\begin{aligned}
    &\tq{For all $i$,} A_i \lequiv{L_i} B_i,
    \\
    &\tq{implies}
    H(A_1,\,\ldots,\,A_n) \lequiv{L_{n + 1}} H(B_1,\,\ldots,\, B_n).
\end{aligned}
\end{equation}
Here $\equiv_L$ denotes equivalence with respect to the logic $L$.
Typically, $n$ is either one or two. For example, writing $\lequiv{FO}$ for equivalence in first-order logic, Mostowski's result is given in the form of~\eqref{eq:fvm-general} as:
\[
    A_1 \lequiv{FO} B_1 \ete{and} A_2 \lequiv{FO} B_2 \ete{implies} A_1 \times A_2 \lequiv{FO} B_1 \times B_2
\]

A key tool in finite model theory for establishing that two structures are logically equivalent is that of a model comparison game.
Examples of such games are the \ef{}~\cite{ehrenfeucht1961application,fraisse1955quelques}, pebbling~\cite{kolaitis1992infinitary} and bisimulation games~\cite{hennessy1980observing}, establishing equivalence in fragments of first order and modal logic.
FVM theorems are typically proven using these games, building a winning strategy for the composite structure out of winning strategies for the components.

Despite their usefulness, working with model comparison games often requires intricate combinatorial arguments that have to be carefully adapted with even the slightest change of problem domain. To mitigate the difficulty with using game arguments directly, finite model theorists introduced higher-level methods such as locality arguments or 0--1 laws to establish model equivalences or inexpressibility.

The recently introduced \emph{game comonad} framework~\cite{AbramskyDW17,abramsky2021relating}
provides a new such tool. This provides a unifying formalism for reasoning about model comparison games, using categorical methods.
The definition of a comonad enables a transfer of results from the formally dual theory of monads, commonly encountered in the categorical semantics of computation and universal algebra \cite{Moggi89,Moggi91,manes2012algebraic}.
Game comonads are designed to encapsulate a particular model comparison game, and are the key abstraction allowing us to give uniform results, whilst avoiding making arguments specific to a particular logic or its corresponding game.

The early work on game comonads focused on capturing many important logics~\cite{AbramskyM21,ConghaileD21,montacute2021pebble}, laying the foundations for further developments. More recent results recover classical theorems from finite model theory, as well as completely new results~\cite{Paine20,DawarJR21,AbramskyMarsden2022,Reggio21poly}. See \cite{Abramsky22survey} for a recent survey and~\cite{AbramskyR21,abramskyreggio2023arboreal} for an axiomatic formulation.

Until now, there has been no account of the composition methods with the game comonad framework.
We close this gap by introducing a novel method for proving FVM theorems within the game comonad approach, giving results uniformly in the classes of structures, operations and logics involved.
Instantiating the abstract results for our example game comonads yields concrete FVM theorems with respect to four typical fragments of first-order logic:
\begin{enumerate}
    \item \emph{the positive existential fragment}, i.e.\ the fragment of first-order logic of formulas not using universal quantification or negations, %
    \item \emph{the existential fragment} i.e.\ the fragment of first-order logic of formulas not using universal quantification, and in which negation can only be applied to atomic formulas.
    \item \emph{counting logic}, i.e.\ the extension of first-order logic with counting quantifiers,
    \item \emph{full} first-order logic.
\end{enumerate}
As is standard in finite model theory~\cite{libkin2004elements,ebbinghaus1999finite}, it is natural to grade these logics by a resource parameter, such as quantifier depth.
Correspondingly, the game comonadic approach to grading is to consider collections of comonads $\C_k$, indexed by a resource parameter $k$.
As a source of concrete examples we shall consider three particular game comonads, $\Ek$, $\Pk$ and $\Mk$.
The pair $\Ek$ and $\Pk$ capture first-order logic, with their resource parameters being quantifier depth and variable count respectively. The comonad $\Mk$ encapsulates modal logic up to modal depth $k$.

Conventionally, proving an FVM theorem involves cleverly combining several winning model-comparison game strategies to form a winning strategy on the composite structures.
In our game comonads approach, we only need to find a collection of morphisms of a specified form to obtain an FVM theorem for the positive existential logic.
This collection of morphisms is a formal witness to a combination of strategies.
Perhaps surprisingly, to witness equivalence in the counting logic, the same collection of morphisms only has to satisfy two additional axioms which coincide with the standard definition of Kleisli law~\cite{manes2007monad}.

Although of practical interest, the FVM theorems in Sections~\ref{s:fvm-positive} and~\ref{s:fvm-counting} are relatively straightforward to prove. We present them in detail for pedagogical reasons to develop the ideas gradually.
The abstract FVM theorem for the full fragment is more technically challenging. This theorem shows that two additional assumptions suffice to extend an FVM theorem from counting logic to the full logic.
We can rephrase these assumptions categorically as showing the operation lifts to a parametric relative right adjoint over an appropriate category.
Surprising connections arise with classical results in monad theory, generalizing the notion of bilinear map and its classifying tensor product from ordinary linear algebra~\cite{kock1971bilinearity,jaklmarsdenshah2022bim}.

Our three running example game comonads $\Ek,\Pk, \Mk$ are introduced in Section~\ref{s:preliminaries}, along with some necessary background on comonads.
Sections~\ref{s:fvm-positive}, \ref{s:fvm-counting} and \ref{s:fvm-standard} all follow the same pattern. We introduce the categorical structure abstracting logical equivalence with respect to the studied logical fragment, establish the corresponding categorical FVM theorem and illustrate this with concrete examples.
The examples in Section \ref{s:fvm-counting} include a new refinement of a result by Dawar, Severini, and Zapata~\cite{dawar2017pebble}, showing that equivalence in 3-variable counting logic without equality and with restricted conjunctions implies cospectrality.
In Section~\ref{s:existential} we observe that techniques from Section~\ref{s:fvm-standard} also automatically entail FVM theorems for the existential fragment.
Section~\ref{sec:product-theorems} develops FVM theorems for products of structures, uniformly in the logic of interest, greatly generalising Mostowski's original result. Unlike our methods, which are parametric in the choice of game comonad, a typical game argument cannot be used to prove such a result, as the game is tied to a specific choice of logic.

In fact, the game comonads $\Ek$ and $\Pk$ naturally deal with first-order logic \emph{without equality}. There is a standard and very flexible technique for incorporating equality, and other ``enrichments'' of logics~\cite{abramsky2021relating}. This is introduced and related to FVM theorems in Section~\ref{s:enrichments}, and illustrated with further examples, including the handling of weak-bisimulation.

\subsection*{Related Work}
Our work can be viewed as comprising model theoretic meta-theorems about the compositionality of logics, parameterised by both the logics and operations of interest.
From this perspective, the original Feferman-Vaught theorem~\cite{feferman1959first, feferman1967first} is essentially a meta-theorem parameterised by a form of generalised product, which can be specialised to yield specific compositionality results specifically for first-order logic. A similar generalised product meta-theorem is used in the setting of modal or temporal logics in~\cite{rabinovich2001compositional, rabinovich2007compositionality}. This work considers specific extensions of modal logic of varying expressive power, but does not go as far as being parameterised by the logics of interest, as in the current work. As well as positive results, this work also considers failure of compositionality, which is outside the scope of the present work.

Surveys of useful compositionality results, for first-order, monadic and guarded second order logic can be found in~\cite{makowsky2004algorithmic, blumencolcolodi2008logicaltheories}.

Preservation of elementary equivalence is studied from a rather different categorical perspective to our own in~\cite{beke2018elementary}. The compatible interaction of composition and observable behaviour is captured by bialgebraic semantics~\cite{turi1997towards}. Although bialgebraic semantics has a rather different intention to the current work, the use of a form of distributive law to witness compositionality bares some intriguing similarity to our results.

Corollary 6.4 in \cite{figueira2023modal} is the only other example of an FVM theorem in the setting of game comonads. It is in fact proved concretely, by making use of the fact that open pathwise-embeddings are precisely strong epimorphisms in the category of the so called \emph{pp-trees}~\cite[Definition~5.4]{figueira2023modal}.
Our work on FVM theorems for game comonads has been motivated by our comonadic account of the Courcelle theorem~\cite{jaklmarsdenshah2022fvm}, where FVM theorems are an essential ingredient.

\section{Preliminaries}
\label{s:preliminaries}

\subsection{Categories of Relational Structures}

We assume familiarity with the basic definitions of category theory, including categories, functors and natural transformations, as can be found in any standard introduction such as~\cite{abramskytzevelekos2010introduction} or \cite{awodey2010category}.
Background on comonads, and any less standard material is introduced as needed.

A \df{relational signature} $\sg$ is a set of \emph{relation symbols}, each with an associated strictly positive natural number \df{arity}.
A~$\sg$-structure~$A$ consists of a~\df{universe} set, also denoted $A$, and for each relation symbol $R \in \sg$ of arity~$n$, an~$n$-ary relation~$R^\As$ on~$A$.
A \df{morphism of~$\sg$-structures} $f\colon \As \rightarrow \Bs$ is a function of type~$A \rightarrow B$, preserving relations, in the sense that for~$n$-ary $R \in \sg$
\[ R^\As(a_1,\ldots,a_n) \qtq{implies} R^\Bs(f(a_1),\ldots, f(a_n)) . \]

For a fixed $\sg$, $\sg$-structures and their morphisms form our main category of interest~$\Rel$. We will also have need of the category of pointed relational structures $\Rels$. Here the objects~$(A,a)$ are a $\sg$-structure $A$ with a distinguished element $a\in A$. The morphisms are $\sg$-structure morphisms that also preserve the distinguished element.

\subsection{Comonads}
\label{s:comonads}
A \df{comonad in Kleisli form} on a category $\CC$ is a triple $(\C, \counit, (-)^*)$ where
\begin{itemize}
    \item $\C$ is an object map $\obj(\CC)\to \obj(\CC)$,
    \item the \df{counit} $\counit$ is an~$\obj({\CC})$-indexed family of morphisms $\counit_A\colon \C(A) \to A$, for $A \in \obj(\CC)$, and
    \item the operation~$(-)^*$ maps morphisms~$f: \C(\As) \rightarrow \Bs$ to their \df{coextension} $f^*:\C(\As) \rightarrow \C(\Bs)$,
\end{itemize}
subject to the following equations:
\begin{align}
    (\counit_A)^* = \id_{\C(A)}, \quad \counit_B \circ f^* = f,\quad (g \circ f^*)^* = g^* \circ f^*,
    \label{eq:comonad-axioms}
\end{align}
It is then standard~\cite{manes2012algebraic} that $\C$ can be equivalently presented in what is referred to as \emph{comonoid form} \((\C, \counit, \delta)\) where \(\C\) is an endofunctor on \(\CC\) and the rest are natural transformation: a \df{counit}~$\counit : \C \Rightarrow \Id$ (where $\Id\colon \CC \to \CC$ is the identity functor) and a \df{comultiplication} $\delta : \C \Rightarrow \C^2$. This triple is required to satisfy $\counit \circ \delta = \id = \C(\counit) \circ \delta$ and $\delta \circ \delta = \C(\delta) \circ \delta$. These equations specify that $(\C,\counit, \delta)$ form a comonoid at an appropriate level of abstraction.
\begin{lem}
    \label{l:comonad-repre}
    Every comonad in Kleisli form \((\C, \counit, (-)^*)\) induces a comonad in comonoid form \((\C, \counit, \delta)\) where \(\C\) is extended to morphisms by setting $\C(f\colon A \to B) = (f \circ \counit_\As)^*$ and \(\delta_A\) is the coextension of the identity morphism~\(\id_{\C(A)}^*\).

    Conversely, a comonad in comonoid form \((\C, \counit, \delta)\) gives a comonad in Kleisli form \((\C, \counit, (-)^*)\) by setting \(f^* = \C(f) \circ \delta_A\).
    These translations are inverse to each other.
\end{lem}

We now introduce the game comonads that we shall regularly refer to in our examples. Each is parameterized by a positive integer~$k$. The first two comonads are defined over $\Rel$.

\begin{description}[leftmargin=0pt, itemindent=!]
    \item[\ef{} comonad~$\Ek$] The universe of~$\Ek(\As)$ consists of non-empty words
    \[
        [a_1, \dots, a_n]
    \]
    with \(a_1,\dots,a_n\) from \(A\) and \(n \leq k\).
    The counit $\counit_A\colon \Ek(A) \to A$ maps $[a_1,\ldots,a_n]$ to $a_n$.

    For an \(m\)-ary relational symbol $R \in \sg$ and \(s_1,\dots,s_m \in \Ek(A)\), set $R^{\Ek(\As)}(s_1,\ldots,s_m)$ iff:
    \begin{enumerate}
        \item the $s_i$ are \emph{pairwise comparable in the prefix order on words}, i.e.\ for every \(i,j\) either \(s_i\) is a prefix of \(s_j\) or vice versa, and
    \item $R^\As(\counit_A(s_1),\ldots,\counit_A(s_m))$.
    \end{enumerate}
    The coextension of~$h : \Ek(\As) \rightarrow \Bs$ is
    \[
        h^*([a_1,\,\ldots,\,a_n])
        \,=\,
        [h([a_1]),\, h([a_1, a_2]),\, \ldots,\, h([a_1,\ldots,a_n])] .
    \]

    \item[Pebbling comonad~$\Pk$] The universe of $\Pk(\As)$ consists of non-empty words
    \[
        [(p_1,a_1),\,\dots,\,(p_n,a_n)]
    \]
    of unrestricted (finite) length such that \(a_1,\dots,a_n \in A\) and \(p_1,\dots,p_n\) are \emph{pebble indexes}, that is, \(p_1,\dots,p_n \in \{0,\dots, k-1\}\).
    The counit $\counit_A$ maps $[(p_1,a_1),\dots,(p_n,a_n)]$ to $a_n$.

    For $m$-ary~$R \in \sg$,
    set $R^{\Pk(\As)}(s_1,\ldots,s_m)$ iff:
    \begin{enumerate}

        \item \label{cond:peb-compare} The $s_i$ are pairwise comparable in the prefix order on words.
        \item \label{cond:peb-active} For~$s_i, s_j$ with~$s_i$ a prefix word of $s_j$, the pebble index of the last pair of~$s_i$ does not reappear in the remaining elements of~$s_j$.
        \item \label{cond:peb-compatible} $R^{\As}(\counit_A(s_1),\ldots,\counit_A(s_m))$.
    \end{enumerate}
    For~$h\colon \Pk(\As) \rightarrow \Bs$, the coextension $h^*$ acts on words as:
    \[
        h^*([(p_1,a_1),\ldots,(p_n,a_n)])
        \,=\,
        [(p_1,h(s_1)), \ldots, (p_n, h(s_n))]
    \]
    where \(s_i = [(p_1,a_1),\dots,(p_i,a_i)]\) for all \(i \in \{1,\dots,n\}\).
\end{description}
We say that a signature $\sg$ is a~\df{modal signature} if it only has relational symbols of arity one or two.
\begin{description}[leftmargin=0pt, itemindent=!]
    \item[Modal comonad~$\Mk$] Defined over $\Rels$ for a modal signature $\sg$, the universe of~$\Mk(A,a_0)$ consists of paths in $A$ starting from $a_0$, i.e.\ sequences
    \[ a_0 \xrightarrow{R_1} a_1 \xrightarrow{R_2} a_2 \xrightarrow{R_3} \ldots \xrightarrow{R_n} a_n \]
    such that~$n \leq k$ and, for every $i$, $R_i\in \sg$ and $R_i^{\As}(a_{i-1},a_i)$. The counit and coextension act as for~$\Ek$.

    For a binary \(R\in \sigma\), the relation $R^{\Mk(A,a_0)}$ consists of the pairs \((s', s)\) where $s$ is the extension of \(s'\) by $a_{n-1} \xrightarrow{R} a_n$ in $\Mk(A,a_0)$. For a unary $P\in \sg$, $P^{\Mk(A,a_0)}(s)$ iff $P^A(a_n)$.
\end{description}

\begin{thmC}[\cite{AbramskyDW17,abramsky2021relating}]
$\Ek$, $\Mk$ and~$\Pk$ are comonads in Kleisli form.
\end{thmC}
The intuition for each of these comonads is that they encode the moves within one structure during the corresponding model comparison game. We shall introduce their mathematical properties as needed in later sections. See~\cite{abramsky2021relating} for detailed discussions of all three comonads.

\section{FVM theorems for positive existential fragments}
\label{s:fvm-positive}
For a comonad $\C$, we introduce the following two relations on structures:
\begin{itemize}
    \setlength\itemsep{0.3em}
    \item $A \parrow{\C} B$ if there exists a morphism $\C(A) \to B$.
    \item $A \pequiv{\C} B$ if $A \parrow{\C} B$ and $B \parrow{\C} A$.
\end{itemize}
It is a standard fact, formulated categorically in~\cite{abramsky2021relating}, that
\begin{align*}
    A \parrow{\Ek} B\qtq{if and only if}
    &\text{for every positive existential sentence $\varphi$ of quantifier }\\
    &\text{depth at most $k$},\  A \models \varphi \text{ implies } B \models \varphi.
\end{align*}
Therefore $A \pequiv{\Ek} B$ if and only if structures $A$ and $B$ agree on the quantifier depth $k$ fragment. Similarly, $A \parrow{\Pk} B$ and $A \pequiv{\Pk} B$ correspond to the same relationships, but for the $k$ variable fragment \cite{AbramskyDW17}.
\begin{rem}
    For conciseness, in this section, and Sections~\ref{s:fvm-counting}, \ref{s:fvm-standard} and \ref{sec:product-theorems}, references to first-order logic implicitly means the variant \emph{without equality} and with infinite conjunctions and disjunctions, as is common in finite model theory. We shall return to the variant with equality in Section~\ref{s:enrichments}.
    Note, however, that typically logical equivalence of two \emph{finite} structures in a fragment of logic with infinite conjunctions and disjunctions is the same as in the corresponding fragment of first-order logic with finite conjunctions and disjunctions.
\end{rem}

We now consider FVM theorems for $\parrow{\C}$ and $\pequiv{\C}$, parametric in $\C$. To motivate our abstract machinery, we first consider a warm-up example, that is nevertheless sufficient to highlight the key ideas.
For a signature $\sg$, and $\tau \subseteq \sg$, the \df{$\tau$-reduct} of a $\sg$-structure $\As$ is the $\tau$-structure on the same universe which interprets all the relation symbols in $\tau$ as in $\As$.
We aim to show that equivalence in positive existential first-order logic with quantifier rank ${\leq}\, k$ is preserved by taking $\tau$-reducts of $\sigma$-structures for any $\tau \subseteq \sigma$.  We can view the $\tau$-reduct operation as a functor
\[ \fgST{(-)} : \Rel \to \R(\tau)\]
which sends a \(\sg\)-structure to its \(\tau\)-reduct and a homomorphism of \(\sg\)-structures \(f\colon A \to B\) is sent to the homomorphism \(\fgST{f}\colon \fgST A \to \fgST B\) with the same underlying function as \(f\) has.

Observe that taking reducts admits an FVM theorem for the \ef{} comonad $\Ek$, viewed as a comonad on both $\Rel$ and $\R(\tau)$ categories. Given the logical reading of $\parrow{\Ek}$, it is an easy observation that $A \parrow{\Ek} B$ implies $\fgST{A} \parrow{\Ek} \fgST{B}$. Indeed, if the positive existential sentences of quantifier depth ${\leq}\,k$ in signature $\sg$ that are true in $A$ are also true in $B$, then the same must hold for the positive existential sentences in the reduced signature $\tau$.

In this case, the FVM theorem can be proved using an ad-hoc argument about the operation involved. However, to identify a general strategy, we would like to prove this fact categorically. Given a morphism $f\colon \Ek(A) \to B$ in $\Rel$ witnessing $A \parrow{\Ek} B$, we wish to construct an $\ol f\colon \Ek(\fgST{A}) \to \fgST{B}$ witnessing $\fgST{A} \parrow{\Ek} \fgST{B}$. To this end, observe that there is a morphism of $\tau$-structures
\[ \kappa_A\colon \Ek(\fgST{A}) \to \fgST{\Ek(A)}\]
which sends a word $[a_1,\dots,a_n]$ in $\Ek(\fgST{A})$ to the same word in $\fgST{\Ek(A)}$. Therefore, $\ol f$ can be computed as the composite
\begin{equation}
    \Ek(\fgST{A}) \xrightarrow{\kappa_A} \fgST{\Ek(A)} \xrightarrow{\fgST{f}} \fgST{B}.
    \label{eq:fgf-ka}
\end{equation}

It is immediate that, since the definition of $\kappa_A$ was not specific to $A$, there is such a morphism for every structure. Consequently, we obtain the following trivial pair of FVM theorems:
\begin{align*}
A \parrow{\Ek} B &\ete{implies} \fgST{A} \parrow{\Ek} \fgST{B} \\
A \pequiv{\Ek} B &\ete{implies} \fgST{A} \pequiv{\Ek} \fgST{B}
\end{align*}

As this exercise demonstrates, all that is required to prove an FVM theorem for the relations $\parrow{\Ek}$ or $\pequiv{\Ek}$, and a unary operation $H$, is to define for every structure $A$ a morphism of $\sg$-structures $\kappa_A\colon \Ek(H(A)) \rightarrow H(\Ek(A))$. It is not difficult to see that the same proof goes through when we parameterise over the comonads involved and allow operations of arbitrary arity. We obtain our first abstract FVM theorem.
\begin{restatable}{thm}{FVMpe}
    \label{t:fvm-pe}
    Let $\C_1,\dots,\C_n$ and $\D$ be comonads on categories $\CC_1,\dots,\CC_n$, and $\CD$ respectively, and
    \[ \op\colon \CC_1\times\dots\times\CC_n \to \CD \]
    a functor.
    If for every $A_1 \in \obj(\CC_1),\, \dots,\, A_n \in \obj(\CC_n)$ there exists a morphism
    \begin{equation}
        \label{eq:kappa-collection}
        \D(H(A_1,\dots,A_n)) \xrightarrow{\!\!\!\overset{\kappa_{A_1,\dots,A_n}}{\quad}\!\!\!} H(\C_1(A_1),\dots,\C_n(A_n))
    \end{equation}
    in $\CD$, then
    \[ A_1 \parrow{\C_1} B_1, \;\ldots,\; A_n \parrow{\C_n} B_n \]
    implies
    \[ H(A_1,\dots,A_n) \parrow{\D} H(B_1,\dots,B_n) \]
    The same result holds when replacing $\parrow{\C}$ with $\pequiv{\C}$.
\end{restatable}
\begin{prf}
    Let $f_i\colon \C_i(A_i) \to B_i$ be the morphism witnessing the relation $A_i \parrow{\C_i} B_i$, for $i\in\{1,\dots,n\}$. Since $\op$ is a functor, the morphism $\ol f$ defined as the composition
    \begin{equation*}
        \D(H(A_1,\dots,A_n)) \xrightarrow{\ \op(f_1,\dots,f_n) \ \circ\ \kappa_{A_1,\dots,A_n}\ } \op(B_1,\dots,B_n)
    \end{equation*}
    witnesses the relation $H(A_1,\dots,A_n) \parrow{\D} H(B_1,\dots,B_n)$.

    The corresponding statement for $\pequiv{\C}$ is obtained by applying the same argument also to the collection of morphisms $g_i\colon \C_i(B_i) \to A_i$ that witness $B_i \parrow{\C_i} A_i$, for $i\in \{1,\dots,n\}$.
\end{prf}

\begin{rem}[Infinitary Operations]
    \label{r:inf-ops}
For readability reasons, we state this as well as the other abstract FVM theorems, such as Theorems~\ref{t:fvm-counting} and~\ref{t:fvm-standard} below, for operations of finite arities only. Nevertheless, all these statements hold verbatim for infinitary operations as well.
\end{rem}

\subsection{Examples}
\label{sec:examples-pe}
\begin{exa}
\label{ex:coproducts-pe}
    As our first simple application of Theorem~\ref{t:fvm-pe}, we obtain an FVM theorem showing that~$\parrow{\Pk}$ is preserved by taking the disjoint union~$A_1 \uplus A_2$ of two $\sg$-structures~$A_1$, $A_2$.
    The universe of~$A_1 \uplus A_2$ can be encoded as pairs~$(i,a)$ where $i \in \{1,2\}$ and~$a \in A_i$.
    For $r$-ary relation $R \in \sg$, the interpretation $R^{A_1 \uplus A_2}$ is defined as
    \[
    R^{A_1 \uplus A_2}((i_1,a_1),\dots,(i_r,a_r))
    \]
    if and only if there exists an $i$ such that for all $1 \leq j \leq r$, $i_j = i$ and
    $R^{A_i}(a_1,\dots,a_r)$.
    We define a morphism of $\sg$-structures
    \[ \kappa_{A_1,A_2}\colon\Pk (A_1 \uplus A_2) \rightarrow \Pk A_1 \uplus \Pk A_2 \]
    which sends a~word~$s = [(p_1,(i_1,a_1)),\dots,(p_n,(i_n,a_n))]$ in~$\Pk (A_1 \uplus A_2)$ to the pair~$(i_n,\nu(s))$,
    where the word~$\nu(s)$ is the restriction of~$s$ to the pairs~$(p_j,a_j)$ such that~$i_j = i_n$.
    Since the definition of~$\kappa_{A_1,A_2}$ is not specific to the structures~$A_1,A_2$, by Theorem~\ref{t:fvm-pe}, the following holds:
    \begin{align*}
    &A_1 \parrow{\Pk} B_1 \ete{and}  A_2 \parrow{\Pk} B_2 \\
    &\tq{implies}
    A_1 \uplus A_2 \parrow{\Pk} B_1 \uplus B_2
    \end{align*}
    Consequently, the same statement holds with $\parrow{\Pk}$ replaced by $\pequiv{\Pk}$.
    The statement demonstrates that if both $A_1$ and $B_1$, and $A_2$ and $B_2$, satisfy the same sentences of the positive existential fragment of $k$-variable logic,
    then the disjoint unions $A_1 \uplus A_2$ and~$B_1 \uplus B_2$ also satisfy the same sentences of that fragment.
    Similar $\kappa$ morphisms can be defined to demonstrate that~$\parrow{\Pk}$, $\pequiv{\Pk}$ and also $\parrow{\Ek}$ and~$\pequiv{\Ek}$ are preserved by taking disjoint unions of \emph{arbitrary} sets of $\sg$-structures.
\end{exa}
Even in the case of relatively simple operations on structures, it may be the case that logical equivalence beyond this fragment is not preserved.
\begin{exa}[Counterexample]
\label{counter:ML-coproducts}
For the modal fragment, a positive modal formula is one without negation or the $\Box$ modality.
As shown in \cite{abramsky2021relating}, for pointed structures
\[ \As = (A,a)\ete{and}\Bs = (B,b)\]
in a modal signature, $\As \parrow{\Mk} \Bs$ if and only if for every positive modal formula $\varphi$ of modal depth at most $k$,
$\As \models \varphi$ implies $\Bs \models \varphi$,
and therefore $\As \pequiv{\Mk} \Bs$ if and only if the two structures agree on positive modal formulae with modal depth bounded by $k$.

Categorically, the disjoint union of~$\sg$-structures is their \emph{coproduct} in~$\Rel$.
The formal definition of coproducts is unimportant for this example, but we note that
coproduct of two \emph{pointed} $\sg$-structures~$(A_1,a_1) + (A_2,a_2)$ in $\Rels$ is a quotient of their disjoint union, where the distinguished points~$a_1$ and $a_2$ are identified.
Unlike with~$\Ek$ and~$\Pk$, the relations~$\parrow{\Mk}$ and~$\pequiv{\Mk}$ are \emph{not} preserved by this operation on pointed~$\sg$-structures, essentially because the two components can interact.
This shows that, although coproduct FVM theorems are relatively simple, it is by no means automatic that they will hold.
\end{exa}

\begin{exa}
\label{ex:coproducts-with-choice-pe}

For every binary relation $R$ in a modal signature $\sg$, there is an operation $\Asan{1} \merge{R} \Asan{2}$ which combines the two pointed structures by adding a new initial point~$\star$ with $R$-transitions to $a_1$ and $a_2$. Operations such as this are commonly found when modelling concurrent systems with process calculi~\cite{stirling2001modal, roscoe2010understanding}. Intuitively this operation avoids the interactions that caused problems in Counterexample~\ref{counter:ML-coproducts}.

There is a morphism of pointed $\sg$-structures
 \[\kappa_{\Asan{1}, \Asan{2}}\colon \comonad{M}_{k+1}(\Asan{1} \merge{R} \Asan{2})\rightarrow \Mk(\Asan{1}) \merge{R} \Mk(\Asan{2}) \]
All sequences in the domain are of the form
\[* \xrightarrow{R} (i,c_0) \xrightarrow{R_1} \ldots \xrightarrow{R_n} (i,c_n), \]
with $i \in \{1,2\}$ and each $(i,c_j) \in A_1 \uplus A_2$ in the same component.
Then, $\kappa_{\Asan{1},\Asan{2}}$ must preserve the distinguished element, and for sequences of length greater than one it sends $* \xrightarrow{R} (i,c_0) \xrightarrow{R_1} \ldots \xrightarrow{R_n} (i,c_n)$ to $(i, c_0 \xrightarrow{R_1} \ldots \xrightarrow{R_n} c_n)$.
    Applying Theorem~\ref{t:fvm-pe} yields the FVM theorem:
\begin{align*}
    &\Asan{1} \parrow{\Mk} \Bsbn{1} \ete{and} \Asan{2} \parrow{\Mk} \Bsbn{2} \\
    &\tq{implies} \Asan{1} \merge{R} \Asan{2} \parrow{\comonad{M}_{k+1}} \Bsbn{1} \merge{R} \Bsbn{2}
\end{align*}
Notice that the resource index actually \emph{increases} from $k$ to $k + 1$ because of the shape of $\kappa$.

\end{exa}

\begin{exa}
\label{ex:morphism-pe}
    A nice feature of the generality of Theorem~\ref{t:fvm-pe} is that for different comonads $\C$ and $\D$ over the same category $\CC$, a natural transformation between the functors $\C \Rightarrow \D$ or, more generally, a collection of morphisms
    \[ \{ \kappa_{A}\colon \D A \rightarrow \C A \mid A \in \obj(\CC)\} \]
    can be seen as instances of \eqref{eq:kappa-collection}, with $H$ equal to the identity functor $\Id\colon \CC \to \CC$.
    This allows us to provide semantic translations for the logics captured by game comonads, a question that had not previously been addressed.
    For example, if we consider $\comonad{P}_2$ as a comonad over $\Rels$ where $\comonad{P}_2(A,a_0)$ has distinguished point $[(0,a_0)]$, then for every $k \in \mathbb{N}$ and object $(A,a_0)$ we define a morphism of pointed $\sg$-structures
    \[ \kappa_{(A,a_0)}\colon \Mk(A,a_0) \rightarrow \comonad{P}_2(A,a_0) \]
    which sends the sequence of transitions
    \[ a_0 \xrightarrow{R_1} \ldots \xrightarrow{R_{n-1}} a_{n-1} \xrightarrow{R_n} a_n \]
    in $\Mk(A,a_0)$
    to the word in $\comonad{P}_2(A,a_0)$ where elements are labeled by the parity of their position
    \[ [(0,a_0),(1,a_1),\dots,(n \bmod 2,a_n)]. \]
    Applying Theorem~\ref{t:fvm-pe} yields the well-known fact that structures which satisfy the same sentences of the positive existential fragment of two-variable logic must also satisfy the same positive modal formulas for any modal depth $k$
    \begin{align*}
    &(A,a_0) \pequiv{\comonad{P}_2} (B,b_0) \\
    &\tq{implies}
    (\forall k \in \mathbb{N}. \ \ (A,a_0) \pequiv{\Mk} (B,b_0)).
    \end{align*}
    Similarly, there exist morphisms
    $\kappa_{A}\colon \Ek A \rightarrow \Pk A$ for every $A$ in $\Rel$ where $[a_1,\dots,a_n]$ with $n\leq k$ is sent to $[(0,a_1),\dots(n-1,a_n)]$ demonstrating that~$\parrow{\Pk}$ refines~$\parrow{\Ek}$.
\end{exa}

\section{FVM theorems for counting logic}
\label{s:fvm-counting}
We now consider another relationship induced by a game comonad. To do so, we first introduce one of the two standard categories associated with any comonad.

\subsection{Kleisli category}
Given a comonad \(\C\) on a category \(\CC\), the \df{Kleisli category}~\cite{kleisli1965every} $\Kl \C$ for $\C$ on $\CC$ is a category which
\begin{itemize}
    \item has the same objects as~$\CC$, and
    \item the morphisms of type $A \rightarrow B$ in $\Kl \C$ are the morphisms of type $\C(A)\to B$ in $\mathcal{C}$. To avoid ambiguity, we shall write $f : A \klto{\C} B$ if we intend $f$ to be understood as a Kleisli morphism.
    \item For an object \(A\), the identity morphism $A \klto{\C} A$ is the counit component $\counit_A : \C(A) \rightarrow A$.
    \item The composite of morphisms $f \colon A \klto \C B$ and $g \colon B \klto \C C$ is given by $g \circ f^*$, where we recall $f^* : \C(A) \rightarrow \C(B)$ is the coextension of~$f \colon \C(A) \to B$. As composition is different to that in the base category, we use the distinct notation $g \kirc f$ for the Kleisli composite of~\(f\) and~\(g\).
\end{itemize}
The axioms for a comonad in Kleisli form ensure that this is a well-defined category. Notice that the morphisms in the Kleisli category are exactly those that were important for positive existential fragments in Section~\ref{s:fvm-positive}.

The motivation for the Kleisli construction is that there is an obvious forgetful functor $U_\C : \Kl{\C} \rightarrow \CC$, and this functor has an identity on objects right adjoint \df{cofree functor} $F_\C : \CC \rightarrow \Kl{\C}$, and this adjunction induces the original comonad.

As with any category, we can consider the isomorphisms in $\Kl \C$. We shall write
\[ A \cequiv{\C} B\]
if structures $A$ and $B$ are isomorphic in the Kleisli category. For finite structures, the relations $A \cequiv{\Ek} B$ and $A \cequiv{\Pk} B$ correspond to equivalence in counting logic \cite{AbramskyDW17,abramsky2021relating}. %
The ideas of Theorem~\ref{t:fvm-pe} can be extended to establish an FVM theorem for counting fragments and the $\cequiv{\C}$ relation. All that is required is to impose extra conditions on the collection of morphisms~\eqref{eq:kappa-collection}. %
Furthermore, these extra conditions turn out to be those of the well-known Kleisli laws related to lifting functors to Kleisli categories~\cite{manes2007monad}.

To motivate our comonadic method, we return to the minimal reduct example from Section~\ref{s:fvm-positive}. We can show that $A \cequiv{\Ek} B$ implies $\fgST{A} \cequiv{\Ek} \fgST{B}$,
via a similar ad-hoc argument to before.
In our setting, we wish to establish that, for mutually inverse morphisms
$f\colon A \klto \Ek B$ and $g\colon B \klto \Ek A$ in $\Kl \Ek$,
the composites $\fgST{f}\circ \kappa_A$ and $\fgST{g} \circ \kappa_B$, constructed as in~\eqref{eq:fgf-ka}, are also mutually inverse in the Kleisli category.
In the following we analyse the abstract setting in which the relation $\cequiv{\C}$ is defined and derive a general result, similar to Theorem~\ref{t:fvm-pe}.

\subsection{Kleisli laws}
\label{s:kleisli-laws}
Akin to our motivating example, we assume a unary operation $H : \CC \rightarrow \CD$,
comonads $(\C, \counit^\C, (-)^*)$ on $\CC$ and $(\D, \counit^\D, (-)^*)$ on $\CD$, and also a collection of morphisms
\[ \{\kappa_A\colon \D H(A) \to H \C(A) \mid A\in \obj(\CC)\}, \]
as in Theorem~\ref{t:fvm-pe}.  We wish to find conditions on this collection of morphisms ensuring that $A \cequiv{\C} B$ implies $H(A) \cequiv{\D} H(B)$. Rephrasing, we require axioms such that if $f\colon A \klto{\C} B$ and $g\colon B \klto{\C} A$ are mutually inverse, then the composites
\[ H(f)\circ \kappa_A\colon H(A) \klto{\D} H(B) \qtq{and} H(g)\circ \kappa_B\colon H(B) \klto{\D} H(A) \]
are also mutually inverse. To do so, we observe that if both
\begin{equation}
    \label{eq:kl-law-coext-one}
    \begin{tikzcd}[column sep=3.3em]
        \D \op(A) \dar[swap]{\kappa_A} \drar{\counit_{\op(A)}} & \\
        \op \C(A) \rar[swap]{H(\counit_A)} & H(A)
    \end{tikzcd}
\end{equation}
and
\begin{equation}
    \label{eq:kl-law-coext-two}
    \begin{tikzcd}[column sep=5em]
        \D \op(A)
            \dar[swap]{\kappa_A}
            \rar{(H(f)\circ \kappa_A)^*}
        &
        \D \op(B)
            \dar{\kappa_B}
        \\
        \op \C(A)
            \rar[swap]{H(f^*)}
        &
        \op \C(B)
    \end{tikzcd}
\end{equation}
commute then, from $g \kirc f = \id$ in $\Kl \D$ (i.e.\ $g \circ f^* = \counit_A$), we obtain that
\begin{align*}
    & (H(g)\circ \kappa_B) \kirc (H(f)\circ \kappa_A) \\
    &= H(g)\circ \kappa_B \circ (H(f)\circ \kappa_A)^*
                                          & \text{definition of \(\kirc\)} \\
    &= H(g)\circ H(f^*) \circ \kappa_A    & \text{\eqref{eq:kl-law-coext-two}} \\
    &= H(g\circ f^*) \circ \kappa_A       & \text{functoriality} \\
    &= H(\counit_A) \circ \kappa_A        & \text{assumption} \\
    &= \counit_{H(A)} : \D(H(A)) \to H(A) & \text{\eqref{eq:kl-law-coext-one}}
\end{align*}
and, similarly,
$f \kirc g = \id$ implies that
\[ (H(f)\circ \kappa_A) \kirc (H(g)\circ \kappa_B) \]
equals to identity in $\Kl \D$.
In fact, equations~\eqref{eq:kl-law-coext-one} and~\eqref{eq:kl-law-coext-two} are equivalent to requiring that the morphisms $\kappa_A$ constitute a natural transformation $\kappa : \D \circ \op \Rightarrow \op \circ \C$ making the following two diagrams commute:
    \begin{equation}
    \begin{tikzcd}
    \D \circ \op \dar[swap]{\kappa} \drar{\counit_{\op}} & \\
    \op \circ \C \rar[pos=0.4]{H\counit} & H
    \end{tikzcd}
    \quad
    \begin{tikzcd}[column sep=1.7em]
    \D \circ \op \dar[swap]{\kappa} \rar{\delta_{\op}} & \D^2 \circ \op \rar{\D\kappa}  & \D\circ\op\circ\C \dar{\kappa_\C} \\
    \op \circ \C \arrow[rr, "H\delta"] & & \op \circ \C^2
    \end{tikzcd}
    \label{eq:klei-nat}
    \end{equation}
Such a $\kappa$ is referred to as a \df{Kleisli law}~\cite{manes2007monad}.
\begin{lem}
    \label{l:klei-ax-clone-form}
    A collection of morphisms $\{\kappa_A\}$ satisfies equations~\eqref{eq:kl-law-coext-one} and~\eqref{eq:kl-law-coext-two} for every $f \colon \C(A) \to B$ iff $\kappa$ is a Kleisli law.
\end{lem}
\begin{proof}
    For the left to right implication, \(\kappa\) is natural since, for any \(g\colon A\to B\), we set \(f = \C(A) \xrightarrow{\counit_A} A \xrightarrow{g} B\) and observe that \(\op(\C(g)) = H(f^*)\) and \(\D(\op(g)) = (H(g) \circ \counit_{H(A)})^* = (\op(f)\circ \kappa)^*\) by \eqref{eq:kl-law-coext-one}. Naturality then follows from \eqref{eq:kl-law-coext-two}. Commutativity of the oblong in~\eqref{eq:klei-nat} follows by setting \(f = \C(A) \xrightarrow{\id} \C(A)\). We then have \((\op(f)\circ \kappa_A)^* = \D(\kappa_A) \circ \delta_{H(A)}\) and \(\kappa_A \circ \D(\kappa_A) \circ \delta_{\op(A)} = \op(f^*) \circ \kappa_A = \op(\delta_A) \circ \kappa_A\) from which the result follows by \eqref{eq:kl-law-coext-one}.

    Conversely, for an \(f\colon \C(A)\to B\), by the oblong law in~\eqref{eq:klei-nat} and naturality of \(\kappa\), \(\op(f^*) \circ \kappa_A = \op(\C(f)) \circ \kappa_{\C(A)} \circ \D(\kappa_A) \circ \delta_{\op(A)} = \kappa_B \circ \D(\op(f)) \circ \D(\kappa_A) \circ \delta_{H(A)} = \kappa_B \circ (\op(f) \circ \kappa_A)^*\).
\end{proof}
\begin{rem}
\label{rem:comonad-morphisms}
    Kleisli laws for monads given by a collection of morphisms satisfying \eqref{eq:kl-law-coext-one} and~\eqref{eq:kl-law-coext-two} appeared in the study of relative monads~\cite{ahrens2012initiality, ahrens2015phd}, under the name \emph{colax morphisms}.
Kleisli laws correspond precisely to liftings (defined as Kleisli liftings below) of the operation $\op\colon \CC \to \CD$ to operations $\Kl \C \to \Kl \D$ commuting with the cofree functors $\CC \to \Kl \C$ and $\CD \to \Kl \D$. See for example~\cite{mulry1993lifting, jacobs1994semantics}.
    Furthermore, a Kleisli law $\kappa : \C \Rightarrow \D$ with respect to the identity functor (as the operation) is the same thing as a~\df{comonad morphism}.
\end{rem}
We will typically be interested in operations of arity greater than one, so we immediately extend our definitions.
\begin{rem}[Notation]
For a set~$I$, and family of categories $\{\CC_i\}_{i \in I}$, we shall write $\veci A$ for a~family of objects, with each $A_i \in \CC_i$, and similarly $\veci f$ for families of morphisms. These form the \df{product category} $\prod_{i \in I} \CC_i$, with composition and identities defined componentwise. Given a family of comonads $\C_i : \CC_i \rightarrow \CC_i$, they induce a comonad $\prod_i \C_i$ on the product category, again componentwise.
\end{rem}
Using this notation, we define the $n$-ary version of the axioms from Lemma~\ref{l:klei-ax-clone-form}. \df{Kleisli laws} for an $n$-ary operation $\op$ are the natural transformations $\kappa\colon \D \circ \op \Rightarrow \op \circ \prod_i \C_i$ such that the following diagrams commute:
\begin{axioms}
    \item[\textbf{\namedlabel{ax:kl-law-counit}{(K1)}}]
    \begin{tikzcd}
        \D \circ \op \dar[swap]{\kappa} \drar{\counit_{\op}} & \\
        \op \circ \prod_i \C_i \rar[swap]{H(\prod_i \counit)} & H
    \end{tikzcd}
    \\[0.5em]
    \item[\textbf{\namedlabel{ax:kl-law-comultiplication}{(K2)}}]
    \begin{tikzcd}[column sep=1.8em]
        \D \circ \op \dar[swap]{\kappa} \rar{\delta_H} & \D^2 \circ \op \rar{\D \kappa} & \D \circ \op \circ \prod_i \C_i \dar{\kappa } \\
        \op \circ \prod_i \C_i \arrow[rr, "\op(\prod_i \delta)"] & & \op \circ \prod_i \C^2_i
    \end{tikzcd}
\end{axioms}
In fact, $n$-ary Kleisli laws are a special case of ordinary (unary) Kleisli laws, as will be shown in Lemma~\ref{lem:n-ary}.

For a family of comonads $\C_i : \CC_i \rightarrow \CC_i$,
we will say that a functor $H': \prod_i \Kl{\C_i} \rightarrow \Kl{\D_i}$
is a \df{Kleisli lifting} of functor $H : \prod_i \CC_i \rightarrow \CD$
if the following diagram commutes:
\begin{equation}
\begin{tikzcd}
    \prod_i \Kl{\CC_i} \rar{H'} & \Kl{\D} \\
    \prod_i \C_i \rar{H} \uar{\prod_i F_{\C_i}} & \CD \uar{F_{\D}}
\end{tikzcd}
\label{eq:kl-lift}
\end{equation}

\begin{lem}
\label{lem:n-ary}
Given a family of comonads $\C_i : \CC_i \rightarrow \CC_i$:
\begin{enumerate}
    \item \label{it:kl-two} The Kleisli construction commutes with products, in that the following equation holds
    \[ \Kl{\prod\nolimits_i \C_i} = \prod\nolimits_i \Kl{\C_i} \]
    \item \label{it:kl-three} The forgetful functor and cofree functors of the product comonad
    \[
        U_{\prod_i \C_i}: \Kl{\prod\nolimits_i \C_i} \rightarrow \prod\nolimits_i \CC_i
        \qquad
        F_{\prod_i \C_i}: \prod\nolimits_i \CC_i \rightarrow \Kl{\prod\nolimits_i \C_i}
    \]
    are given by the product \(\prod_i U_{\C_i}\) of forgetful and product \(\prod_i F_{\C_i}\) of cofree functors, respectively.
    \item \label{it:kl-four} An n-ary Kleisli laws of type $\D \circ H \Rightarrow H \circ \prod_i \C_i$ is the same thing as a (unary) Kleisli law involving the product comonad $\prod_i \C_i$.
    \item \label{it:kl-five} n-ary Kleisli laws bijectively correspond to Kleisli liftings of $H$.
\end{enumerate}
\end{lem}
\begin{prf}
    Items \ref{it:kl-two} and \ref{it:kl-three} are routine calculations from the definitions.
Item \ref{it:kl-four} follows directly from the definitions,
and so we can reduce the bijective relationship between n-ary Kleisli laws and Kleisli liftings to the well-known unary case. See for example~\cite{mulry1993lifting}. This establishes item~\ref{it:kl-five}.
\end{prf}

The argument presented for a unary opereration generalises smoothly to operations of any arity, yielding an FVM theorem for counting fragments.
\begin{restatable}{thm}{FVMcounting}
    \label{t:fvm-counting}
    Let $\C_1,\dots,\C_n$ and $\D$ be comonads on categories $\CC_1,\dots,\CC_n$ and $\CD$, respectively, and
    $ \op\colon \prod_i \CC_i \to \CD $
    a functor.
    If there exists a Kleisli law of type
    \( \D\circ H \Rightarrow H \circ \prod\nolimits_i\C_i \)
    then
    \[ A_1 \cequiv{\C_1} B_1 \;\ldots,\; A_n \cequiv{\C_n} B_n \]
    implies
    \[ H(A_1,\dots,A_n) \;\cequiv{\D}\; H(B_1,\dots,B_n) . \]
\end{restatable}
\begin{prf}
    Assume \(A_1 \cequiv{\C_1} B_1\), \enspace\ldots,\enspace \(A_n \cequiv{\C_n} B_n\) or, equivalently,
\[ F_{\C_1}(A_1) \cong F_{C_1}(B_1), \enspace\ldots,\enspace F_{\C_n}(A_n) \cong F_{\C_n}(B_n). \]
    The assumed Kleisli law induces a Kleisli lifting \(H'\colon \prod_i \Kl{\C_i} \rightarrow \Kl{\D_i}\) by item \ref{it:kl-five} of Lemma~\ref{lem:n-ary}, which by functoriality implies
\[ H'(F_{\C_1}(A_1), \ldots, F_{\C_n}(A_n)) \cong H'(F_{C_1}(B_1),\ldots, F_{\C_n}(B_n)) . \]
By the commutativity of \eqref{eq:kl-lift} we have
\( F_{\D}(H(A_1, \ldots, A_n)) \cong F_{\D}(H(B_1,\ldots, B_n))\),
which encodes the required property.
\end{prf}

\subsection{Examples}

\begin{exa}
    \label{ex:coprod-etc-counting}
The collection of $\kappa$ morphisms described in each of
    Examples~\ref{ex:coproducts-pe},~\ref{ex:coproducts-with-choice-pe} and~\ref{ex:morphism-pe} of Section~\ref{sec:examples-pe}
    are natural and satisfy axioms~\ref{ax:kl-law-counit} and~\ref{ax:kl-law-comultiplication}.
Therefore, by applying Theorem~\ref{t:fvm-counting} to the Kleisli laws in
    Examples~\ref{ex:coproducts-pe} and~\ref{ex:coproducts-with-choice-pe}
        we have that $\cequiv{\Ek}$ and $\cequiv{\Pk}$ are preserved by taking coproducts of structures and the merge operation maps $\cequiv{\Mk}$ to $\cequiv{\comonad{M}_{k+1}}$.
Recalling Remark~\ref{rem:comonad-morphisms}, by applying Theorem~\ref{t:fvm-counting} to the comonad morphism in Example~\ref{ex:morphism-pe} we have that for all $k \in \mathbb{N}$, ~$\cequiv{\comonad{P}_2}$ refines $\cequiv{\Mk}$, and $\cequiv{\Pk}$ refines $\cequiv{\Ek}$.
\end{exa}

\begin{exa}[Cospectrality]
Two graphs $G$ and $H$ are \df{cospectral} if the adjacency matrices of $G$ and $H$ have the same multiset of eigenvalues.
We shall use Theorem~\ref{t:fvm-counting} to strengthen a result in \cite{dawar2017pebble} showing that equivalence in 3-variable counting logic \emph{with equality} implies cospectrality.

Since cospectrality is a notion on graphs, we move to the full subcategory of loopless undirected graphs in signature $\sg = \{E\}$ where $E$ is a binary edge relation.
The comonad $\Pk$ restricts to this category and we also consider a comonad $\Cos$ defined in~\cite[Section 3]{abramskyjaklpaine2022discrete} which characterizes cospectrality,
in that
$G \cequiv{\Cos} H$ is equivalent to $G$ and $H$ being cospectral.
The universe of $\Cos(G)$ can be encoded as pairs~$(c,v_i)$ where $c$ is a closed walk\footnote{Recall that a \df{walk} is a sequence of (possibly repeated vertices) where $v_j$ is adjacent to $v_{j+1}$.} $v_0 \xrightarrow{E} v_1 \xrightarrow{E} \ldots \xrightarrow{E} v_n \xrightarrow{E} v_0$ that passes through $v_i$.
Two pairs $(c,v_i),(c',v_j)$ are adjacent in $\Cos(G)$ if $c = c'$ and $v_i$ is adjacent to $v_j$ in the closed walk $c$.
The counit $\counit_G$ maps $(c,v_i)$ to $v_i$.
For $h\colon \Cos(G) \rightarrow H$, the coextension $h^{*}$ maps the pair $(c,v_i)$ with $c$ being the closed walk $v_0 \xrightarrow{E} v_1 \xrightarrow{E} \ldots \xrightarrow{E} v_n \xrightarrow{E} v_0$ to $(d,h(c,v_i))$ where $d$ is the closed walk $h(c,v_0) \xrightarrow{E} h(c,v_1)\xrightarrow{E} \ldots \xrightarrow{E} h(c,v_n) \xrightarrow{E} h(c,v_0)$.

We define a comonad morphism $\kappa\colon\Cos(G) \to \comonad{P}_3(G)$ where $(c,v_i)$ is mapped to the word
\[ [(2,v_0),(1,v_1),(0,v_2),(1,v_3),\dots,(i \bmod 2,v_i)]. \]
Recalling Remark~\ref{rem:comonad-morphisms}, we apply Theorem~\ref{t:fvm-counting}, deducing that $\cequiv{\comonad{P}_3}$ implies $\cequiv{\Cos}$.
Since $\cequiv{\comonad{P}_3}$ captures equivalence in 3-variable counting logic \emph{without equality} and $\cequiv{\Cos}$ captures cospectrality, we have avoided the need for equality in the logic.
This fact is also a consequence of \cite[Theorem 32]{DawarJR21} and the original theorem of Dawar, Severini, and Zapata~\cite{dawar2017pebble}.
However, the same reasoning allows us to define a comonad morphism $\Cos \Rightarrow \comonad{PR}_3$ where $\comonad{PR}_3$ is the pebble-relation comonad from~\cite{montacute2021pebble}, capturing the restricted conjunction fragment of $3$-variable counting logic.
The universe of $\comonad{PR}_3(G)$ consists of pairs $(s,i)$ with $s \in \comonad{P}_3(G)$ is a sequence of length $n$ and $i \in \{0,1,\dots,n\}$ is an index into the sequence $s$.
In this case, $(c,v_i)$ is mapped to the pair $(s,i)$ where $s = \kappa(c,v_n)$.
Intuitively, $s$ enumerates the entire closed walk $c$ and the index $i$ picks out $v_i$ in this walk.
This proves a genuine strengthening of the previous result.
\end{exa}

\begin{exa}
    The theses \cite[Section~6.3.3]{shah2024thesis} and the forthcoming \cite{karamlou2026thesis} take a look at the quantum monad \(\mathbf Q_d\) from \cite{abramsky2017quantum} on categories of structures with a distinguished `commeasurablity' relation.
    For each $d > 0$, $\mathbf Q_d(A)$ is a structure on $d$-dimensional projector-valued measurements of \(A\), which is used for modeling perfect quantum strategies in a non-local homomorphism game.
    Similarly, one can adapt the \(\Ek\) comonad to the same category.
    It is shown in the \emph{op.\ cit.} (independently) that there is a Kleisli law \(\Ek \circ \mathbf Q_d \Rightarrow \mathbf Q_d \circ \Ek\), giving us that the equivalence in the counting fragment \(A \lequiv{\#\Ek} B\) implies \(\mathbf Q_d(A) \lequiv{\#\Ek} \mathbf Q_d(B)\).
\end{exa}

\section{Coalgebras and open pathwise-embeddings}
\label{s:fvm-standard}

In Sections~\ref{s:fvm-positive} and \ref{s:fvm-counting} we described FVM theorems for the equivalence relations $\pequiv{\C}$ and $\cequiv{\C}$, typically expressing logical equivalence for the positive existential and counting logic variants. To do the same for the full logic requires us to move from the Kleisli category to the richer setting of the \df{Eilenberg--Moore category of coalgebras}.

For a comonad $\C$ on $\CC$, a pair $(A,\alpha)$ where $\alpha\colon A \to \C(A)$ is a morphism in $\CC$ is a \df{$\C$-coalgebra} or simply just \df{coalgebra} if the following two diagrams commute.
\begin{equation}
    \begin{tikzcd}
        A \dar[swap]{\alpha} \ar{rd}{\id} \\
        \C(A) \rar[pos=0.40]{\counit_A} & A
    \end{tikzcd}
    \qquad
    \qquad
    \begin{tikzcd}
        A \dar[swap]{\alpha} \rar{\alpha} & \C(A) \dar{\delta_A} \\\
        \C(A) \rar{\C(\alpha)} & \C(\C(A))
    \end{tikzcd}
    \label{eq:axioms-coalg}
\end{equation}

A \df{morphism of coalgebras} $f\colon (A,\alpha) \to (B,\beta)$ is a morphism $f\colon A\to B$ such that $\beta \circ f = \C(f) \circ \alpha$. We write $\EM \C$ for the Eilenberg--Moore category of coalgebras of $\C$.

The original motivation for the Eilenberg-Moore construction is that there is an obvious forgetful functor $\EMU \C : \EM{\C} \rightarrow \CC$, mapping \((A,\alpha)\) to \(A\), and it has a right adjoint \df{cofree coalgebra} functor $\EMF\C : \CC \rightarrow \EM{\C}$, such that the adjunction $\EMU\C \dashv \EMF\C$ induces the comonad~$\C$. We will need the cofree construction in later developments. Concretely its action on objects is $\EMF\C(\As) = (\C(\As), \delta_\As)$.

\begin{rem}
    A coalgebra $(A,\alpha)$ of any of the game comonads defined in Section~\ref{s:comonads} is endowed with a preorder~$\sqsubseteq_\alpha$, where $x \sqsubseteq_\alpha y$ whenever the word $\alpha(x)$ is a prefix of $\alpha(y)$.
    A \df{forest order} is a poset in which the set of predecessors of any element in a finite chain. It follows from the coalgebra axioms in~\eqref{eq:axioms-coalg} that $\sqsubseteq_\alpha$ is a forest order on the universe. Furthermore, one can check that \(\sqsubseteq_\alpha\) is \df{compatible} with the relational structure of \(A\). This means that \(R^A(a_1,\dots,a_n)\) implies that for every \(i,j\) either \(a_i \sqsubseteq_\alpha a_j\) or \(a_j \sqsubseteq_\alpha a_i\).

    In fact $\EM{\Ek}$ is equivalent to the category of $\sg$-structures endowed with a compatible forest order of depth ${\leq}\,k$. Similar characterisations of the categories $\EM{\Pk}$ and $\EM{\Mk}$ can be made as well, cf.~\cite[Section 9]{abramsky2021relating}.
    \label{r:coalg-order}
\end{rem}
\smallskip
We are now almost ready to define the relation $\sequiv \C$.
Given $A,B\in \CC$, define $A \sequiv{\C} B$
if there exists a span of \emph{open pathwise-embeddings}
\[ \EMF \C (A) \leftarrow Z \rightarrow \EMF \C (B) . \]
We postpone the definition of open pathwise-embeddings to Section~\ref{s:ope-preservation}.

In terms of our example comonads, $A \sequiv{\Ek} B$ and $A \sequiv{\Pk} B$ correspond to agreement on first-order sentences of $k$-bounded quantifier depth and variable count respectively \cite{abramsky2021relating}. Similarly, $(A,a) \sequiv{\Mk} (B,b)$ characterises agreement on modal formulae of modal depth at most $k$.

Returning again to the example of reducts, showing the trivial fact that $A \sequiv \Ek B$ implies $\fgST{A} \sequiv \Ek \fgST{B}$ in our setting suggests the strategy:
\begin{enumerate}
    \item \label{en:lift}
    Lift an $n$-ary operation $H$ to the level of coalgebras, in a manner that commutes with the construction $\EMF \C (-)$ of cofree coalgebras.
    \item \label{en:preservation} Check that open pathwise-embeddings are preserved by the lifted operation.
\end{enumerate}
Tackling step~\ref{en:lift} is the topic of Section~\ref{s:lifting-operations}. Step~\ref{en:preservation} is the subject of the subsequent Section~\ref{s:ope-preservation}.

\subsection{Lifting operations to coalgebras}
\label{s:lifting-operations}

Here we focus on the task of lifting a given operation to the level of coalgebras.
By an \df{Eilenberg--Moore lifting}, or just \df{lifting} in short, of a functor $\op : \prod_i \CC_i \rightarrow \CD$ we mean a functor
$\lop : \prod_i \EM{\C_i} \rightarrow \EM{\D}$ such that the following diagram commutes
up to isomorphism:
\begin{equation}
    \begin{tikzcd}
        \prod_i \EM{\C_i} \rar{\lop} & \EM{\D} \\
        \prod_i \CC_i \uar{\prod_i \EMF{\C_i}} \rar{\op} & \CD \uar{\EMF{\D}}
    \end{tikzcd}
    \label{eq:lifting}
\end{equation}
To this end, we generalise a standard technique from monad theory, used for lifting adjunctions to the Eilenberg--Moore categories \cite{johnstone1975adjoint} or to lifting monoidal structure to the algebras of commutative monad \cite{kock1971closed} (see also~\cite{jacobs1994semantics,seal2013,jaklmarsdenshah2022bim}).
For such a lifting to exist, given a Kleisli law, it is sufficient that the codomain category of coalgebras has sufficient equalisers.
In our case, it is enough to check that the comonad on the target category preserves embeddings and the rest is ensured by Linton's theorem~\cite{linton1969coequalizers}.
In the following paragraphs we review the terminology necessary for the said theorem.

We say that a pair of classes of morphisms \((\Qu, \Em)\) is a \df{weak factorisation system} in a category \(\CC\) if:
\begin{itemize}
    \item Every morphism \(f\) in \(\CC\) factorises as \(f = m \circ e\) for some \(m\in \Em\) and \(e \in \Qu\).
    \item \(\Qu = \{ e \mid \forall m\in \Em, e \pitchfork m\}\) and \(\Em = \{ m \mid \forall e \in \Qu, e \pitchfork m\}\), where \(e \pitchfork m\) expresses that \(e\) is weakly orthogonal to \(m\), that is, for any commutative diagram as shown on the left-hand side below
    \[
        \begin{tikzcd}
            \bullet \rar{e} \dar & \bullet \dar \\
            \bullet \rar{m} & \bullet
        \end{tikzcd}
        \qquad
        \qquad
        \begin{tikzcd}
            \bullet \rar{e} \dar & \bullet \dar \ar[swap]{ld}{d} \\
            \bullet \rar{m} & \bullet
        \end{tikzcd}
    \]
    there exists a diagonal morphism \(d\) as shown in the diagram on the right-hand side above, making the diagram commute.
\end{itemize}
Morphisms of \(\Qu\) are called \df{quotients} and are denoted by \(\quot\) and morphisms of \(\Em\) are \df{embeddings} which we denote by \(\emb\).
We further say that \((\Qu,\Em)\) is a \df{proper factorisation system} if it is a weak factorisation system such that $\Qu$ is a class of epimorphisms and $\Em$ is a class of monomorphisms.

\begin{exa}
    Factorisation systems are omnipresent in category theory. For example, the category of sets is equipped with the (surjectives, injectives) proper factorisation system. Similarly, \(\Rel\) and \(\Rels\) come equipped with \((\text{epis}, \text{strong monos})\) factorisation system \cite{adamek1994locally}. In this case epis are precisely the surjective homomorphisms of (pointed) $\sg$-structures and strong monos are precisely the \df{embeddings}, that is, injective homomorphisms $f\colon A\to B$ which reflect relations.
    In fact, strong monos are the same as regular, extremal, or effective monos in this \(\Rel\) and \(\Rels\).

    In the following, we always assume that (pointed) \(\sg\)-structures come equipped with (epi, strong mono) factorisation system.
\end{exa}

Finally, for a comonad \(\C\) on a category \(\CC\) with a proper factorisation system we say that \(\C\) \df{preserves embeddings} if $\C(f)$ is an embedding whenever $f$ is. In similar fashion, an $n$-ary operation $\op$ preserves embeddings if $\op(f_1,\dots,f_n)$ is an embedding whenever $f_1,\dots,f_n$ are embeddings in their respective categories.

\begin{thm}[Linton \cite{linton1969coequalizers}]
        \label{t:linton}
        Given a comonad \(\C\) on a category \(\CC\) equipped with a proper factorisation system \((\Qu, \Em)\) if
        \begin{enumerate}
            \item $\C$ preserves $\Em$-morphisms,
            \item $\CC$ has coproducts and
            \item \(\CC\) is \df{$\Em$-well-powered}, that is, every object \(A\) of \(\CC\) has, up to isomorphism, only a set of embeddings \(X \emb A\).
        \end{enumerate}
        then $\EM \C$ has equalisers.
\end{thm}
\begin{rem}
This is essentially the dual of~\cite[Proposition 4]{linton1969coequalizers}, up to some small technical differences. We provide a full proof in the form we require for completeness.
\end{rem}

To prove Linton's theorem, we first establish some preliminaries. We need a standard fact that comonads that preserve embeddings lift the factorisation system from the underlying category to their category of coalgebras. This is essential not only for the proof of Theorem~\ref{t:linton} but also for the definition of~\(\lequiv\C\) in Section~\ref{s:ope-preservation}.

\begin{lem}
    \label{l:lift-fs}
    Let \(\C\) be a comonad on a category \(\CC\) with a proper factorisation system \((\Qu,\Em)\). If \(\C\) preserves embeddings then \(\EM\C\) admits a proper factorisation system $(\ol{\Qu}, \ol{\Em})$ where a morphism of coalgebras $h \colon (A,\alpha) \to (B,\beta)$ is in $\ol{\Qu}$ (resp.\ $\ol{\Em}$) if the underlying morphism $h \colon A \to B$ is in~$\Qu$ (resp.~$\Em$).
\end{lem}
\begin{proof}
    The statement is analogous to that of Lemma 2 in \cite{linton1969coequalizers}, except that we do not require functoriality of factorisations which, however, plays no role in the proof.
\end{proof}

Next, we prove an intermediate statement from which Linton's theorem follows. In the proof of Theorem~\ref{t:linton} the category \(\CA\) will be \(\EM\C\) for some comonad \(\C\).

\begin{lem}
    \label{l:linton}
    If a category \(\CA\) has a proper factorization system \((\Qu,\Em)\), is \(\Em\)-well-powered and has coproducts then \(\CA\) has equalisers.
\end{lem}
\begin{proof}
    Consider parallel pair:
    \[
    \begin{tikzcd}
    \As
    \rar[yshift=0.5em]{f}
    \rar[swap,yshift=-0.5em]{g}
    & \Bs
    \end{tikzcd}
    \]
    As~$\CA$ is well-powered and has a factorization system, there is a set~$I$ of isomorphism classes~\([m_i]\) of embeddings~$m_i : \struct{M}_i \rightarrow \As$ equalizing~$f$ and~$g$. As
    \[ f \circ [m_i] = [f \circ m_i] = [g \circ m_i] = g \circ [m_i] \]
    the induced morphism \(\wt m : \coprod_i M_i \to A\) equalises \(f\) and \(g\).
    As~$\CA$ has a factorization system, the morphism~$\wt m$ factors as:
    \[
    \begin{tikzcd}
    \coprod_i \struct{M}_i \rar[twoheadrightarrow]{e} & \struct{E} \rar[rightarrowtail]{m} & A
    \end{tikzcd}
    \]
    Next, assume~$h : \Cs \rightarrow \As$ equalises~$f$ and~$g$. Observe that \(m'\) in the \((\Qu,\Em)\) factorisation \(m'\circ e'\) of \(h\) also equalises \(f\) and \(g\) since \(e'\) is an epimorphism. We may assume that \(m ' = m_i\) for some~$i$.
    Let~$j$ be the corresponding coproduct injection. The following commutes:
    \[
    \begin{tikzcd}
    \coprod_i \struct{M}_i \rar[twoheadrightarrow]{e} & \struct{E} \rar[rightarrowtail]{m} & \As \rar[yshift=0.5em]{f} \rar[swap,yshift=-0.5em]{g}& \Bs \\
    \struct{M}_i \uar{j} & \Cs \lar{e'} \urar[swap]{h} & &
    \end{tikzcd}
    \]
    As the factorization system is proper, $m$ is a monomorphism, and so~$e \circ j \circ e'$ is the unique morphism \(C \to E\) factorising \(h\) via \(m\).
    Consequently, \(m\) equalises \(f\) and \(g\).
\end{proof}

We're now ready to prove Linton's theorem.
\begin{prf}[Proof of Theorem~\ref{t:linton}.]
    We check that \(\EM\C\) satisfies the conditions of Lemma~\ref{l:linton}.
    By Lemma~\ref{l:lift-fs}, since \(\C\) preserves embeddings the proper factorisation system \((\Qu,\Em)\) on \(\CC\) lifts to a proper factorisation system \((\ol\Qu, \ol\Em)\) on \(\EM\C\). By the definition of \(\ol\Em\) since \(\CC\) is well-powered so is \(\EM\C\). Finally, \(\EM\C\) has coproducts since the forgetful functor \(\EM\C \to \CC\) reflects them.
\end{prf}

All the comonads that we have introduced thus far preserve embeddings and \(\Rel\) and \(\Rels\) both have coproducts. In the case of \(\Rel\) these are disjoint unions of structures and in the case of \(\Rels\) we glue the disjoint union by the distinguished elements. Finally, we see that both categories are also well-powered because each embedding \(X \emb A\) is determined by its image in \(A\) and there are only set-many of these.

\begin{cor}
    \label{c:equalisers}
    If a comonad $\C$ on $\Rel$ or on $\Rels$ preserves embeddings then $\EM{\C}$ has equalisers.
    In particular, \(\EM\Ek\), \(\EM\Pk\) and \(\EM\Mk\) have equalisers.
\end{cor}

Let us come back to the question of lifting functors $\op\colon \prod_i\CC_i\to \CD$ to the category of coalgebras, in sense of \eqref{eq:lifting} above. It might come as a surprise that nothing more than a Kleisli law and existence of equalisers is needed for the lift to exist.

\begin{restatable}{thm}{EMLifting}
    \label{t:em-lifting}
    If $\kappa$ is a Kleisli law of type $\D \circ \op \Rightarrow \op \circ \prod_i \C_i$ and $\EM\D$ has equalisers then the lifting of $\op$ to $\lop : \prod_i \EM{\C_i} \rightarrow \EM{\D}$ exists.
\end{restatable}

\begin{rem}[Notation]
    In the following we often just write $\alpha$ instead of $(A,\alpha)$ in expressions that involve a coalgebra $(A,\alpha) \in \EM \C$. In particular, the morphism $\iota_{\vec{\alpha_i}}$ below is parametrised by the coalgebras $\{\, (A_i,\alpha_i)\,\}_i$ and the same is true for the resulting coalgebra~$\lop(\vec{\alpha_i})$.

    Furthermore, we often omit the subscripts for the components of natural transformations, if they can be easily inferred from the context. For example, we shall write \(\kappa\) instead of \(\kappa_{A_1,\dots,A_n}\) and $\lambda_H$ instead of the more verbose $\lambda_{H(A_1,\dots,A_n)}$.
\end{rem}

\begin{prf}
    Let \((A_i,\alpha_i) \in \EM{\C_i}\) be a collection of coalgebras, for \(i \in I\).
    By our assumptions, we can form the following equaliser in~\(\EM\D\).
    \begin{equation}
        \begin{tikzcd}[column sep=2.5em]
            \lop(\vec{\alpha_i})
                \rar{\iota_{\vec{\alpha_i}}}
            & \EMF{\D}(\op(\veci A))
                \ar[yshift=0.5em]{rr}{\EMF{\D}(\kappa)\circ \delta}
                \ar[swap,yshift=-0.5em]{rr}{\EMF{\D}(\op(\vec{\alpha_i}))}
            &
            & \EMF{\D}(\op(\vec{\C_i(\As_i)}))
        \end{tikzcd}
        \label{eq:equaliser-functor}
    \end{equation}
    Observe that the mapping on objects $\vec{\alpha_i} \mapsto \lop(\vec{\alpha_i})$ extends to a functor. Given coalgebra morphisms $f_i \colon (A_i, \alpha_i) \to (B_i, \beta_i)$, for $i\in \{1,\dots,n\}$, obtain a diagram of the following form.
    \begin{equation}
        \begin{tikzcd}[column sep=2.5em]
            \lop(\vec{\alpha_i})
                \dar[dashed,swap]{\lop(f_i)}
                \rar{\iota_{\vec{\alpha_i}}}
            & \EMF{\D}(\op(\vec{\As_i}))
                \dar[swap]{\EMF{\D}(\op(\veci f))}
            & \\
            \lop(\vec{\beta_i})
                \rar[swap]{\iota_{\vec{\beta_i}}}
            & \EMF{\D}(\op(\vec{\Bs_i}))
                \rar[yshift=0.5em]{\EMF{\D}(\kappa)\circ \delta}
                \rar[swap,yshift=-0.5em]{\EMF{\D}(\op(\vec{\beta_i}))}
            & \EMF{\D}(\op(\vec{\C_i(\Bs_i)}))
        \end{tikzcd}
        \label{eq:hatOp-functorial}
    \end{equation}
    It follows from the fact that $f_1,\dots,f_n$ are coalgebra morphisms that $\EMF{\D}(\op(\veci f)) \circ \iota_{\vec{\alpha_i}}$ equalises the parallel morphisms at the bottom, i.e.\ that $\EMF{\D}(\kappa)\circ \delta \circ \EMF{\D}(\op(\veci f)) \circ \iota_{\vec{\alpha_i}} = \EMF{\D}(\op(\vec{\beta_i})) \circ \EMF{\D}(\op(\veci f)) \circ \iota_{\vec{\alpha_i}}$. Then, by the universal property there exists a unique morphism $\lop(\vec{\alpha_i}) \to \lop(\vec{\beta_i})$, which we denote by $\lop(f_i)$, and which makes the square above commute.

    To verify that \(\lop\) is a lifting of \(\op\), let \(A_i\in \CC_i\) be a collection of objects, for \(i\in I\), and consider the following diagram.
    \begin{equation}
        \begin{tikzcd}[column sep=5.8em]
            \EMF\D(\op(\veci A))
                \rar{\EMF\D(\kappa)\mm\circ \delta_H}
            & \EMF{\D}(\op(\vec{\C_i(A_i)}))
                \ar[yshift=0.5em]{r}{\EMF{\D}(\kappa_{\veci \C})\mm\circ \delta_{H\veci \C}}
                \ar[swap,yshift=-0.5em]{r}{\EMF{\D}(\op(\vec{\delta_{A_i}}))}
            & \EMF{\D}(\op(\vec{\C^2_i(\As_i)}))
        \end{tikzcd}
        \label{eq:equaliser-cofree}
    \end{equation}
    We show that this is the equaliser diagram \eqref{eq:equaliser-functor} for the coalgebras \(\{\EMF{\C_i}(A_i)\}_i\).

    First, we check that \(\EMF\D(\kappa)\circ \delta_H\) equalises \(\EMF\D(\kappa_{\veci \C})\circ \delta_{H\veci \C}\) and \(\EMF{\D}(\op(\vec{\delta_{A_i}}))\) by a short computation:
    \begin{align*}
        &\D(\kappa_{\veci \C}) \circ \delta_{H\veci \C} \circ \D(\kappa) \circ \delta_H         \\
        &= \D(\kappa_{\veci \C}) \circ \D^2(\kappa) \circ \delta_{\D \op} \circ \delta_H   &\text{naturality of \(\delta\)} \\
        &= \D(\kappa_{\veci \C}) \circ \D^2(\kappa) \circ \D(\delta_H) \circ \delta_H &\text{comonad law for \(\delta\)} \\
        &= \D\op(\vec{\delta_{A_i}}) \circ \D(\kappa) \circ \delta_H & \text{Kleisli law \eqref{eq:klei-nat}}
    \end{align*}

    Furthermore, assume \(f\colon (V,\nu) \to \EMF{\D}(\op(\vec{\C_i(A_i)}))\) also equalises \(\EMF\D(\kappa_{\veci \C})\mm\circ \delta_{H\veci \C}\) and \(\EMF{\D}(\op(\vec{\delta_{A_i}}))\), then \(\D\op(\vec{\counit_{A_i}}) \circ f\colon (V,\nu) \to \EMF\D(\op(\veci A))\) factors \(f\) via \(\EMF\D(\kappa)\circ \delta_H\) as:
    \begin{align*}
        & \D(\kappa) \circ \delta_H \circ \D\op(\vec{\counit}) \circ f \\
        &= \D(\kappa) \circ \D^2(\op(\vec{\counit})) \circ \delta_H \circ f & \text{naturality of \(\delta\)}\\
        &= \D(\op(\vec{\C_i(\counit_{A_i})})) \circ \D(\kappa_{\veci \C}) \circ \delta_H \circ f  & \text{naturality of \(\kappa\)} \\
        &= \D(\op(\vec{\C_i(\counit_{A_i})})) \circ \D(\op(\vec{\delta_{A_i}})) \circ f  & \text{\(f\) equalises \eqref{eq:equaliser-cofree}}\\
        &= f  &\text{comonad law for \(\delta\) and \(\counit\)}
    \end{align*}

    Uniqueness of this factorisation follows from the fact that \(\EMF\D(\kappa)\circ \delta_H\) is a monomorphism. In fact, it is a split monomorphism, since \(\D\op(\vec{\counit_{A_i}})) \circ \D(\kappa) \circ \delta_H = \D(\counit_{\op(\veci A)}) \circ \delta_H = \id\).
\end{prf}

\begin{rem}
    Observe that \eqref{eq:equaliser-functor} is a coreflexive equaliser, i.e.\ it is an equaliser of a parallel pair of morphisms with a common retraction (post-inverse).  The common retraction of the parallel pair in \eqref{eq:equaliser-functor} is given by \(\D H(\vec{\counit_{A_i}})\). Therefore, the theorem holds in greater generality with the category~\(\EM\D\) only required to have equalisers of coreflexive pairs.
\end{rem}

\begin{rem}
    Appendix A of \cite{balan2011coalgebras} contains a statement analogous to our Theorem~\ref{t:em-lifting} above. In fact the authors show that in case of a unary functor \(\op\colon \CC\to\CD\), the lifting \(\lop\colon \EM\C \to \EM\D\) is the right Kan extension of \(K_\D\circ \op'\) along \(K_\C\) where \(\op'\) is the Kleisli lifting \(\Klei\C \to \Klei\D\) from \eqref{eq:kl-lift} and \(K_\C\colon \Klei\C \to \EM\C\) is the comparison functor.
\end{rem}

\subsection{Open pathwise-embeddings and their preservation}
\label{s:ope-preservation}

In this section we define open pathwise-embeddings and give sufficient conditions for the lifted functors of Section~\ref{s:lifting-operations} to preserve them. In fact, our conditions are a combination of two other categorical properties: parametric adjoints and relative adjoints, which we discuss in Appendix~\ref{s:parametric-relative-adjoints}.

To emphasise the main ingredients of our arguments we restrict to a simplified setting of \df{path categories}. These are categories \(\CA\) equipped with a proper factorisation system \((\Qu, \Em)\) and a choice of \df{path} objects \(\Pa \sue \obj(\CA)\).

\begin{wrapfigure}[5]{r}{6.5em}
    \begin{tikzcd}
        P\dar[>->,swap]{e} \rar[>->]{g} & Q \dar[>->]{m} \ar[dashed,swap]{ld}{d} \\
        X \rar[swap]{f} & Y
    \end{tikzcd}
\end{wrapfigure}
In this setting, an embedding $e\colon P \emb X$ is a \df{path embedding} if \(P\) is a path (i.e.\ \(P\in \Pa\)). Then, \(f\colon X\to Y\) is \df{pathwise-embedding} if $f\circ e$ is an embedding for every path embedding $e\colon P \embed X$. Finally, $f$ is \df{open} if for every
commutative square of solid arrows as shown on the right,
where $e,m,g$ are path embeddings,
there exists a morphism $d\colon Q \to X$ such that $e = d \circ g$ and $m = f \circ d$.

\begin{exa}
    Recall from Remark~\ref{r:coalg-order} that coalgebras $(A, \alpha)$, for either of our example comonad \(\C \in \{\Ek, \Pk,\Mk\}\), come equipped with a forest order $\sqsubseteq_\alpha$ on the universe. In the following, we always assume that paths \(\Pa \sue \EM\C\) are the coalgebras $(A,\alpha)$ which are finite linear orders in $\sqsubseteq_\alpha$.

    Also, since \(\C\) typically preserves embeddings of \(\sg\)-structures, \(\EM\C\) admits a proper factorisation system \((\ol\Qu, \ol\Em)\) by lifting the (epi, strong mono) factorisation system of (pointed) \(\sg\)-structures, cf.\ Lemma~\ref{l:lift-fs}.
    Consequently, we view \(\EM\Ek\), \(\EM\Pk\), and \(\EM\Mk\) as path categories in the above sense.
    In fact, every \emph{arboreal category}, in sense of \cite{AbramskyR21}, is a path category. Note that the notion of path category in \cite{AbramskyR21} is different from ours.

    Furthermore, coalgebra morphisms \(f\colon (A, \alpha) \to (B,\beta)\) in \(\EM\C\), for any comonad~\(\C\) from \(\{\Ek, \Pk,\Mk\}\), are homomorphisms of the underlying structures \(A \to B\) which respect the forest orders. The latter means that \(f\) is a poset morphism (i.e.\ \(a \sqsubseteq_\alpha a'\) implies \(f(a) \sqsubseteq_\beta f(a')\)) which, furthermore, sends elements of height 1 (i.e.\ roots) to elements of height 1, elements of height 2 to elements of height 2, and so on. In case of \(\C = \Pk\) or \(\C = \Mk\) the additional tree decorations such as the pebble index or binary relations need to be preserved too.
    Now, \(f\) is a pathwise-embedding precisely whenever any restriction of \(f\) along a finite branch is an embedding. Finally, \(f\) is open whenever \(f\) is a bounded morphism (in sense of modal logic) with respect to the tree orders.
\end{exa}

\begin{rem}
With the notion of open pathwise-embedding for $\EM \C$ with paths and embeddings specified as above, we recall that
$A \sequiv \C B$ iff there exists a pair of open pathwise-embeddings $\EMF \C(A) \leftarrow Z \rightarrow \EMF \C(B)$, where $\EMF \C(A)$ is the cofree coalgebra on $A$.
\end{rem}

We now come back to our original goal. In order for an $n$-ary operation $\op \colon \prod_i \CC_i \to \CD$ to admit an FVM theorem, we wish to show that its lifting
$\lop\colon \prod\nolimits_i \EM{\C_i} \to \EM \D$
described in Theorem~\ref{t:em-lifting}, sends tuples of open pathwise-embeddings to open pathwise-embeddings.

In our simplified setting, we have a functor \(\lop\colon \prod\nolimits_i \CA_i \to \CB\) between path categories.
In the remainder of this section show that \(\lop\) preserves open pathwise-embeddings assuming the following two conditions:

\begin{axioms}
    \item[\textbf{\namedlabel{ax:s1}{(S1)}}]
    $\lop\colon \prod_i \CA_i \to \CB$ preserves embeddings.

    \item[\textbf{\namedlabel{ax:s2}{(S2)}}]
    Any path embedding $e\colon P \embed \lop(\veci A)$ has a \df{minimal decomposition} as $e_0\colon P \to \lop(\veci P)$ followed by $\lop(\veci e)\colon \lop(\veci P) \to \lop(\veci A)$,
         for some path embeddings $e_i\colon P_i \embed A_i$, for $1 \leq i \leq n$.
\end{axioms}
Minimality in \ref{ax:s2} expresses that for any decomposition of $e$ as $g_0 \colon P \to \lop(\veci Q)$ followed by $\lop(\veci g) \colon \lop(\veci Q) \to \lop(\veci A)$,
for some path embeddings $g_i\colon Q_i \embed A_i$, there exist necessarily unique morphisms $h_i\colon P_i \to Q_i$ such that $e_i = g_i \circ h_i$, for $i=1,\dots,n$.

We are going to need the following standard facts about embeddings of a proper factorisation system. The dual statement for quotients holds too but we do not need it here.

\newpage %
\begin{lem}[e.g.\ Section 2 in \cite{freyd1972categories}]
    \label{l:fs-basics}
    For a category with a proper factorisation system~\((\Qu,\Em)\):
\begin{enumerate}
    \item All equalisers are embeddings.
    \item Embeddings are closed under composition.
    \item If $g \circ f$ is an embedding then so is $f$.
\end{enumerate}
\end{lem}

These basic properties of embeddings listed in Lemma~\ref{l:fs-basics} allow us to show the first half of our desired result, preservation of pathwise-embeddings.

\begin{restatable}{lem}{PEPreserved}
    If $f_1,\dots,f_n$ are pathwise-embeddings in $\CA_1$, \dots, $\CA_n$, respectively, then so is $\lop(f_1,\dots,f_n)$.
    \label{l:pe-preserved}
\end{restatable}
\begin{prf}
    Assume $e\colon P \embed \lop(\veci A)$ is a path embedding where, for each $i$, $A_i$ is the domain of~$f_i$. By~\ref{ax:s2}, $e$ decomposes as:
    \[
        \begin{tikzcd}[column sep=3.5em]
            P \rar{e_0}
            & \lop(\veci P)
            \rar{\lop(\veci e)}
            & \lop(\veci A)
        \end{tikzcd}
        \]
    Observe that, by Lemma~\ref{l:fs-basics}.3, $e_0$ is an embedding because $e$ is. Further, since $f_i$ is a pathwise-embedding, for $i=1,\dots,n$, the morphism $f_i \circ e_i$ is an embedding. Therefore, by \ref{ax:s1}, $\lop(\vec{f_i \circ e_i})$ is also an embedding. We obtain that the composite $\lop(\veci f) \circ e = \lop(\vec{f_i \circ e_i}) \circ e_0$ is an embedding because embeddings are closed under composition, cf. Lemma~\ref{l:fs-basics}.2.
\end{prf}

For the preservation of openness, we need the following technical lemma.

\begin{restatable}{lem}{LemmaIFour}%
    \label{l:I4}
    Let $f_i\colon A_i \to B_i$ in $\CA_i$ be pathwise-embeddings, for each $i$, and let $e$ and $g$ be path embeddings making the diagram on the left below commute.
    \[
    \begin{tikzcd}[ampersand replacement=\&, column sep=1.0em]
        \& P \ar[>->,swap]{dl}{e} \ar[>->]{dr}{g} \\
        \lop(\veci A) \ar{rr}{\lop(\veci f)} \& \& \lop(\veci B)
    \end{tikzcd}
    \quad
    \begin{tikzcd}[ampersand replacement=\&]
        \Ps_i \rar{f'_i}\dar[swap,>->]{e_i} \& \Qs_i\dar[>->]{g_i}\\
        \As_i \rar{f_i} \& \Bs_i
    \end{tikzcd}
    \]
    Then, for $i=1,\dots,n$, there exist morphisms
    \[ f'_i\colon P_i \to Q_i ,\]  such that
    the diagram on the right above commutes.
    Here, the $e_i$ and $g_i$ are the minimal embeddings given by \ref{ax:s2} such that $e$ and $g$ decompose through $\lop(\veci e)$ and $\lop(\veci g)$, respectively.
\end{restatable}
\begin{prf}
    Let $e_0\colon P \embed \lop(\veci P)$ be the embedding such that
    $ \lop(\veci e) \circ e_0 $
    is the minimal decomposition of $e$ by \ref{ax:s2}. Since
    \[ \lop(\veci f) \circ e = \lop(\vec{f_i \circ e_i}) \circ e_0 \]
    is another decomposition of $g$, there exists morphisms $l_i\colon Q_i \embed P_i$ such that $g_i = f_i \circ e_i \circ l_i$, for $i=1,\dots,n$.

    Next, we observe that $e_0 = \lop(\veci l) \circ g_0$ where
    \[ g_0\colon P \embed \lop(\veci Q) \]
    is such that $\lop(\veci g) \circ g_0$ is the minimal decomposition of $g$. Observe that
    \begin{align*}
    \lop(\vec{f_i \circ e_i}) \circ e_0
    = g
    = \lop(\veci g) \circ g_0
    = \lop(\vec{f_i \circ e_i \circ l_i}) \circ g_0
    = \lop(\vec{f_i \circ e_i}) \circ \lop(\veci l) \circ g_0 .
    \end{align*}
    Therefore, since by \ref{ax:s1} the morphism $\lop(\vec{f_i \circ e_i})$ is an embedding and since embeddings are monomorphisms, we obtain  that $e_0 = \lop(\veci l) \circ g_0$.

    Consequently, we obtain another decomposition of $e$, given by $\lop(\vec{e_i \circ l_i}) \circ g_0$. By minimality $\veci e$, there exists
    \[ f'_i\colon P_i \to Q_i, \]
    for $i=1,\dots,n$, such that $e_i = e_i \circ l_i \circ f'_i$. Therefore, for $i=1,\dots,n$, we see that
    \[ f_i \circ e_i = f_i \circ e_i \circ l_i \circ f'_i = g_i \circ f'_i, \]
    which concludes the proof.
\end{prf}

We are now ready to proof the main theorem of this section.

\begin{restatable}{thm}{SmoothnessTheorem}
    \label{t:smoothness}
    If $\lop$ satisfies \ref{ax:s1} and \ref{ax:s2} and $f_1,\dots,f_n$ are open pathwise-embeddings, then so is $\lop(\veci f)$.
\end{restatable}
\begin{prf}
    Given for $i=1,\dots,n$, open pathwise-embeddings $f_i\colon A_i \to B_i$ in $\EM{\C_i}$, it is enough to check that $\lop(\veci f)$ is open by Lemma~\ref{l:pe-preserved}. Assume that the outer square of path embeddings in the diagram below commutes, with the left-most and right-most morphisms being their minimal decompositions by \ref{ax:s2}.
    \[
        \begin{tikzcd}[column sep=3.5em] %
            P
                \ar[>->]{rr}{h}
                \ar[>->,swap]{d}{e_0}
                \ar[>->]{dr}{e'_0}
            &
            & Q
                \ar[>->]{d}{g_0}
            \\
            \lop(\veci P)
                \ar[swap,>->]{d}{\lop(\veci e)}
                \ar[dashed,swap]{r}{\lop(\vec{f'_i})}
            & \lop(\veci{P'})
                \ar[sloped,>->,swap]{rd}{\lop(\vec{e'_i})}
                \rar[dashed]{\lop(\veci h)}
            & \lop(\veci Q)
                \ar[>->]{d}{\lop(\veci g)}
            \\
            \lop(\veci A)
                \ar{rr}{\lop(\veci f)}
            &
            & \lop(\veci B)
        \end{tikzcd}
    \]
    The path embedding $\lop(\vec{f_i\circ e_i}) \circ e_0$ has a minimal decomposition going via
    $\lop(\vec{e'_i})$ %
    as shown above.
    Then, by minimality of this decomposition and by Lemma~\ref{l:I4}, for $i=1,\dots,n$, there exist $f'_i\colon P_i \to P'_i$ and $h_i\colon P'_i \to Q_i$ such that $f_i \circ e_i = e'_i \circ f'_i$ and $e'_i = g_i \circ h_i$. Since $f_i$ is open and $f_i \circ e_i = g_i \circ h_i \circ f'_i$, there is a morphism $d_i\colon Q_i \to A_i$ such that $d_i \circ h_i \circ f'_i = e_i$ and $g_i = d_i \circ f_i$. Finally, because the outer rectangle and the bottom rectangle commute and $\lop(\veci g)$ is a monomorphism, %
    the top rectangle commute as well. Consequently, $\lop(\veci d)\circ g_0\colon Q \to \lop(\veci A)$ is the required diagonal filler of the outer square.
\end{prf}

\section{FVM theorems for the full logic}

Theorems~\ref{t:em-lifting} and \ref{t:smoothness} above can be used to state an abstract FVM theorem for full fragments.
However, in order to use such theorem in concrete instances we need to be able to check its assumptions easily.
To this end, we rephrase axioms \ref{ax:s1} and \ref{ax:s2} in more elementary terms.
The key ingredient is the notion of \bimorph{} which allows us to understand morphisms involved in axiom \ref{ax:s2} in terms of the original operation and underlying categories.

\subsection{\Bimorph{}s}
Recall from linear algebra that, for vector spaces \(U,V,W\), a~bilinear map (or bimorphism) is a function \(f\colon U\times V \to W\) such that, for any \(u\in U\) and \(v\in V\) the functions \(f(u,-)\) and \(f(-,v)\) are linear. The universal property of bimorphisms states that bilinear maps \(U\times V \to W\) are in a bijective correspondence with (ordinary) linear maps \(U \otimes V \to W\), where \(\otimes\) denotes the tensor product operation on vector spaces.
This is an instance of a general phenomenon from the theory of monads, that is, the duals of comonads. Concretely, if a monad \(T\) is \emph{commutative} (i.e.\ there is a Kleisli law \(T(A) \times T(B) \to T(A\times B)\) satisfying extra conditions~\cite{kock1971closed}) then the category of algebras for \(T\) has a tensor operation \(\otimes\), which is the lifting of the product operation \(\times\) (in dual sense to our Section~\ref{s:lifting-operations}). Moreover, algebra homomorphisms \(A\otimes B \to C\) correspond to certain bilinear morphisms \(A\times B \to C\)~\cite{kock1971bilinearity}.

In what follows, we often need to analyse morphisms of the form $(A,\alpha) \to \lop(\vec{(B_i,\beta_i)})$, for some coalgebras $(A,\alpha),(B_1,\beta_1),\dots,(B_n,\beta_n)$. To aid with this we restate a direct generalisation of the universal property of bilinear maps and~\cite{kock1971bilinearity} to our setting.
Intuitively, instead of algebra homomorphisms \(A\otimes B \to C\) we study coalgebra morphisms \((A,\alpha) \to \lop(\vec{(B_i,\beta_i)})\) and instead of bilinear morphisms \(A\times B \to C\) we describe the corresponding morphisms \(A \to \op(\vec{B_i})\), which we call \emph{\bimorph}\footnote{The name originates in~\cite{ellerman2006theory} where it was used to talk about morphisms between objects of different categories. These usually take the form \(A \to F(B)\) where \(F\) is a functor. In our case we further require the condition \eqref{eq:bimorph} enforced by the Kleisli law.}.

In this subsection we fix a Kleisli law \(\kappa : \D \circ \op \Rightarrow \op \circ \prod_i \C_i\) and its corresponding lifting $\lop : \prod_i \EM{\C_i} \rightarrow \EM{\D}$ given by Theorem~\ref{t:em-lifting}. %
Given a morphism of coalgebras \(f\colon (A,\alpha) \to \lop(\vec{(B_i,\beta_i)})\) we define a morphisms \(f^\#\colon A \to \op(\veci B)\) in \(\CD\) as
\[
    f^\#\colon  A \xrightarrow{\EMU\D(f)} \EMU\D(\lop(\vec{B_i,\beta_i})) \xrightarrow{\univ_{\veci \beta}} H(\veci B)
\]
where \(\univ_{\veci \beta}\) is the composition \(\counit_\op \circ \EMU\D(\iota_{\veci \beta})\) of the underlying morphism of
\(\iota_{\veci \beta} \colon \lop(\veci \beta) \to \EMF \D(\op(\veci B))\)
with the counit morphism
\( \counit_\op \colon \D(H(\veci B)) \to H(\veci B)\).

In the following we observe that \(f^\#\) is a \df{(Kleisli) \bimorph{}}, that is, a morphism \(g\colon A \to H(\veci B)\), from \((A,\alpha)\) to \(\vec{(B_i,\beta_i)}\), written as \(g : \alpha \to [\veci \beta]\), such that the following diagram commutes.
    \begin{equation}
        \begin{tikzcd}[->, ampersand replacement=\&,]
        \As \arrow[rr, "g"] \dar[swap]{\alpha} \& \& \op(\vec{\Bs_i}) \dar{\op(\vec{\beta_i})} \\
        \D(\As) \rar{\D(g)} \& \D\op(\vec{\Bs_i}) \rar{\kappa} \& \op(\vec{\C_i(\Bs_i)})
    \end{tikzcd}
    \label{eq:bimorph}
    \end{equation}

In fact, Proposition~\ref{p:bimorphisms} below establishes that the mapping $f \mapsto f^\#$ gives a bijection between coalgebra morphisms and \bimorph{}s.
First we observe that \bimorph{}s are closed under composition by coalgebra morphisms from both left and right.

\begin{lem}
    \label{l:bimorph-composition}
    Given a \bimorph{} $f\colon \alpha \to [\veci \beta]$, a morphism of $\D$-coalgebras $h\colon (A',\alpha') \to (A,\alpha)$ and $\C_i$-coalgebra morphisms $g_i\colon (B_i, \beta_i) \to (B'_i, \beta'_i)$, for $i\in \{1,\dots,n\}$, the composite $\op(\veci g) \circ f \circ h$ is a \bimorph{} $\alpha' \to [\veci{\beta'}]$.
    \qed
\end{lem}

Crucially we also need that $\univ_{\veci \beta}\colon \EMU\D(\lop(\vec{B_i,\beta_i})) \to \op(\veci B)$ from above is a universal \bimorph{}, in the following sense.
\begin{lem}
    The following holds for $\univ_{\veci \beta}$.
    \begin{enumerate}
        \item $\univ_{\veci \beta}$ is a \bimorph{} $\lop(\veci \beta) \to [\veci \beta]$.
        \item For any \bimorph{} $f \colon \alpha \to [\veci \beta ]$ there is a \emph{unique} $\D$-coalgebra morphism $f^+ \colon (A,\alpha) \to \lop(\veci \beta )$ such that
        \( f = \univ_{\veci \beta} \circ f^+ \)
        in the underlying category $\CD$.
        \item The collection of morphisms $\univ_{\veci \beta}$ is natural in $\veci \beta$. That is, for any tuple of coalgebra morphisms $f_i\colon \beta_i \to \beta_i'$, for $i\in \{1,\dots,n\}$,
        \[
           \begin{tikzcd}
               \EMU\D(\lop(\vec{B_i,\beta_i})) \dar[swap]{\EMU \D (\lop(\veci f))} \rar{\univ_{\veci \beta}} & \op(\veci B) \dar{\op(\veci f)} \\
               \EMU\D(\lop(\vec{B'_i,\beta'_i})) \rar{\univ_{\veci \beta'}} & \op(\veci B')
           \end{tikzcd}
        \]
    \end{enumerate}
    \label{l:bimorph-basics}
\end{lem}
\begin{prf}
    To simplify the notation we denote the underlying morphism of \(\iota_{\veci \beta}\) by \(\iota_{\veci \beta}\) and by \((W,\omega)\) we denote the \(\D\)-coalgebra \(\lop(\veci \beta)\), i.e.\ \(W = \EMU\D(\lop(\veci \beta))\) and \(\omega \colon W \to \D(W)\) is such that \(\D(\iota_{\veci \beta})\circ \omega = \delta^\D_{\op(\veci B)} \circ \iota_{\veci \beta}\).
    Item 1 is a routine calculation:
    \begin{align*}
        & \kappa \circ \D(\counit_H) \circ \D(\iota_{\veci \beta}) \circ \omega \\
        &= \kappa \circ \D(\counit_H) \circ \delta_{\op(\veci B)} \circ \iota_{\veci \beta} & \text{\(\iota_{\veci \beta}\) is a coalgebra homomorphism} \\
        &= \kappa \circ \counit_{\D H} \circ \delta_{\op(\veci B)} \circ \iota_{\veci \beta} & \text{counit axioms} \\
        &= \counit_{H\C} \circ \D\kappa \circ \delta_{\op(\veci B)} \circ \iota_{\veci \beta} & \text{naturality of \(\counit\)} \\
        &= \counit_{H\C} \circ \D(\op(\veci \beta)) \circ \iota_{\veci \beta} & \text{\(\iota_{\veci \beta}\) coequalises \eqref{eq:equaliser-functor}} \\
        &= \op(\veci \beta) \circ \counit_{H} \circ \iota_{\veci \beta} & \text{naturality of \(\counit\)}
    \end{align*}
    Item 2 follows the equaliser universal property. First, recall that the underlying morphism $f \colon \alpha \to [\veci \beta]$ uniquely determines the coalgebra morphism \(\D(f) \circ \alpha\colon (A,\alpha) \to \EMF\D(\op{\veci B})\) such that \(f = \counit_\op \circ \D(f) \circ \alpha\).
    Furthermore, this coalgebra morphism equalises~\eqref{eq:equaliser-functor}.
    \begin{align*}
        & \D\kappa \circ \delta_{H(\veci B)} \circ \D(f) \circ \alpha \\
        &= \D\kappa \circ \D^2(f) \circ \delta_{A} \circ \alpha & \text{naturality of \(\delta\)} \\
        &= \D\kappa \circ \D^2(f) \circ \D(\alpha) \circ \alpha & \text{\(\alpha\) is a morphism \((A,\alpha) \to \EMF\D A\)} \\
        &= \D(\op(\veci \beta)) \circ \D(f) \circ \alpha & \text{\bimorph{} assumption \eqref{eq:bimorph}}
    \end{align*}
    From \(\iota_{\veci \beta}\) being equaliser of \eqref{eq:equaliser-functor}, there is a unique \(f^+\colon (A,\alpha) \to (W,\omega)\) such that \(\iota_{\veci \beta} \circ f^+ = \D(f) \circ \alpha\). Consequently, \(\univ_{\veci \beta} \circ f^+ = \counit_\op \circ \D(f) \circ \alpha = f \circ \counit \circ \alpha = f\).

    Item 3 follows from the fact that $\counit$ is a natural transformation and from the definition of $\lop(\veci f)$ in terms of the universal property of equalisers, cf.\ \eqref{eq:hatOp-functorial} above.
\end{prf}

We are now ready to prove the promised statement.
\begin{restatable}{prop}{Bimorphisms}
    \label{p:bimorphisms}
    The mapping $f \mapsto f^\#$ establishes a one-to-one correspondence between
    \begin{itemize}
        \item coalgebra morphisms \((A,\alpha) \to \lop(\vec{(B_i,\beta_i)})\) and
        \item \bimorph{}s \(\alpha \to [\veci \beta]\).
    \end{itemize}
\end{restatable}
\begin{prf}
    Let $f\colon (A,\alpha) \to \lop(\vec{(B_i, \beta_i)})$ be a coalgebra morphism.
    Since $\univ_{\veci \beta}$ is a \bimorph{} (by Lemma~\ref{l:bimorph-basics}.1), so is $f^\#$ by Lemma~\ref{l:bimorph-composition}. Furthermore, $(f^\#)^+ = f$ by the universal property of $\univ_{\veci \beta}$ (Lemma~\ref{l:bimorph-basics}.2). On the other hand, for a \bimorph{} $g \colon \alpha \to [\veci \beta]$ we clearly have $(g^+)^\# = g$, again by the universal property of $\univ_{\veci \beta}$.
\end{prf}

Lastly, we show that the mapping $f \mapsto f^\#$ commutes with compositions of coalgebra morphisms.
\begin{prop}
    \label{p:bimorph-dist}
    For a $\D$-coalgebra morphism $h\colon (A',\alpha') \to (A,\alpha)$ and $\C_i$-coalgebra morphisms $g_i\colon (B_i, \beta_i) \to (B'_i, \beta'_i)$:
    \[
        (\lop(\veci g) \circ f \circ h)^\# = \op(\veci g) \circ f^\# \circ h.
    \]
\end{prop}
\begin{prf}
    Let $f \colon (A,\alpha) \to \lop(\vec{(B_i,\beta_i)})$ and $h \colon (A',\alpha') \to (A,\alpha)$ be $\D$-coalgebra morphisms. Then,
    \[
        (f \circ h)^\# = \univ_{\veci \beta} \circ f \circ h = f^\# \circ h.
    \]
    Similarly, for coalgebra morphisms $g_i \colon (B_i, \beta_i) \to (B'_i, \beta'_i)$, for $i\in \{1,\dots,n\}$, we have
    \[
        (\lop(\veci g) \circ f)^\# = \univ_{\veci{\beta'}} \circ \EMU\D(\lop({\veci g})) \circ f = \op(\veci g) \circ \univ_{\veci \beta} \circ f = \op(\veci g) \circ f^\#
    \]
    where the equality in the middle follows from Lemma~\ref{l:bimorph-basics}.3. We see that $(\lop(\veci g) \circ f \circ h)^\# = \op(\veci g) \circ f^\# \circ h$ now follows from these two observations and Lemma~\ref{l:bimorph-composition}, which ensures that $f \circ h$ is a \bimorph{}.
\end{prf}

\subsection{Simplifying the axioms}
\label{s:simpler-axioms}

Axioms~\ref{ax:s1} and \ref{ax:s2} give us sufficient conditions for functors between arbitrary path categories to preserve open pathwise-embeddings.
In our case the path categories of interests are categories of coalgebras for our game comonads and the functor in question \(\lop\) is typically a lifted functor \(\op\) on the base categories. In this section we describe more elementary conditions, written in terms of \bimorph{}s, that imply axioms~\ref{ax:s1} and \ref{ax:s2}.
This way we can easily establish the validity of assumptions of our FVM theorems for full logic (Theorem~\ref{t:fvm-standard} below), without having to analyse the lifted functor $\lop$.

In this section we fix a functor between base categories \(\op\colon \prod_i \CC_i \to \CD\) and assume that it lifts to $\lop\colon \prod_i \EM{\C_i} \to \EM\D$ as in Theorem~\ref{t:em-lifting}.
To begin with, we observe that \ref{ax:s1} is implied by a similar property of $\op$.

\begin{restatable}{prop}{SOnePrime}
    \label{p:s1p}
    If $\op$ preserves embeddings then so does~$\lop$.
\end{restatable}
\begin{prf}
    Let
    \[e_1\colon (A_1,\alpha_1) \to (B_1,\beta_1), \ \dots, \ e_n \colon (A_n,\alpha_n) \to (B_n,\beta_n) \]
    be embeddings in $\EM{\C_1},\dots,\EM{\C_n}$, respectively. Recall that $\lop(e_1,\dots,e_n)$ is defined by the universal property of equalisers. In particular, it makes the following diagram commute.
    \[
        \begin{tikzcd}
            \lop(\vec{\Ac_i})\dar[swap]{\lop(\veci e)} \rar{\iota_{\vec{\Ac_i}}} & \EMF{\D}(\op(\vec{\As_i})) \dar{\EMF{\D}(\op(\vec{U(e_i)}))} \\
            \lop(\vec{\Ac'_i}) \rar{\iota_{\vec{\Ac'_i}}} & \EMF{\D}(\op(\vec{\As'_i}))
        \end{tikzcd}
    \]
    From $e_1,\dots,e_n$ being embeddings we know that $\EMF{\D}(\op(\vec{U(e_i)}))$ is an embedding as well because $\D$ and \(\op\) preserve embeddings. Furthermore, by Lemma~\ref{l:fs-basics}.1, $\iota_{\vec{\Ac_i}}$ is an embedding since it is an equaliser and, consequently, by Lemma~\ref{l:fs-basics}.3, $\lop(\veci e)$ is an embedding too.
\end{prf}

To proceed we need the following simple lemma about decompositions of \bimorph{}s.

\begin{lem}
    \label{l:multi-decomp}
    Assume that both $\D$ and $\op$ preserve embeddings.
    Then, given a \bimorph{} $f\colon \Ac \to [\vec{\Bc_i}]$ which decomposes as $f_0\colon \As \to \op(\vec{\As_i})$ followed by $\op\left(\veci e\right)$, for some coalgebra embeddings $e_i\colon \Ac_i \embed \Bc_i$, for $1 \leq i \leq n$, the morphism $f_0$ is a \bimorph{} $\Ac \to \left[\vec{\Ac_i}\right]$.
\end{lem}
\begin{prf}
    We show that the left oblong in the diagram below commutes.
    \[
    \begin{tikzcd}[column sep=3.5em]
        \As
            \rar{f_0}
            \ar{dd}{\alpha}
        & \opveci\As
            \rar{\opveci e}
            \ar{d}{\opveci\alpha}
        & \opveci\Bs
            \ar{d}{\opveci \beta}
        \\
        & \opvec{\D(\As_i)}
            \rar{\opvec{\D(e_i)}}
        & \opvec{\D(\Bs_i)}
        \\
        \D(P)
            \rar{\D(f_0)}
        & \D(\opveci \As)
            \uar{\kappa}
            \rar{\D(\opveci e)}
        & \D(\opveci \Bs)
            \uar{\kappa}
    \end{tikzcd}
    \]
    Since $f$ is a \bimorph{}, we know that the outer square commutes. Further, the two squares on the right commute by naturality of $\kappa$ and the assumption that $e_i$, for $1\leq i \leq n$, is a coalgebra morphism. A simple diagram chasing implies that
    \[ \opvec{\D(e_i)} \circ \opveci\alpha \circ f_0 = \opvec{\D(e_i)}\circ \kappa \circ \D(f_0)\circ \alpha .\]
    Since $\D$ and $H$ preserve embeddings, $\opvec{\D(e_i)}$ is an embedding and hence a mono, %
    proving that $f_0$ is a \bimorph{}.
\end{prf}

Next we observe that \ref{ax:s2} is implied by a simple condition which can be stated entirely in terms of \bimorph{}s, that is, typically, it is a condition in the base category of (pointed) relational structures.

\begin{restatable}{prop}{STwoPrime}
    \label{p:s2p}
    \label{p:smoothness-bimorphisms}
    The axiom \ref{ax:s2} is implied by the following axiom:
    \begin{axioms}
        \item[\em\textbf{\namedlabel{ax:s2p}{(S2')}}]
        For any path $(\Ps,\Pc)$ in $\EM \D$ and coalgebras $(\As_i,\Ac_i)$ in $\EM{\C_i}$ for $1 \leq i \leq n$,
        every \bimorph{} \(\Pc \to [\veci \Ac]\)
        has a minimal decomposition through
        \[ \op(\veci e)\colon \op(\vec{\Ps_i}) \to \op(\vec{\As_i}),\]
        for some path embeddings $e_i\colon (\Ps_i,\Pc_i) \embed (\As_i,\Ac_i)$.
    \end{axioms}
\end{restatable}
\begin{prf}
    Let $e\colon (P,\pi) \embed \lop(\vec{(\As,\Ac_i)})$ be a path embedding and let $f = e^\#$ be be the corresponding homomorphism of (pointed) $\sg$-structures $P \to H(\vec{\As_i})$, obtained by Proposition~\ref{p:bimorphisms}. %

    By \ref{ax:s2p}, $f$ has a minimal decomposition as $\op(\veci e)\circ f_0$ for some $f_0\colon \Ps \to \op(\vec{\Ps_i})$ and path embeddings
    \[ e_i\colon (\Ps_i,\Pc_i) \embed (\As_i,\Ac_i), \]
    with $1 \leq i \leq n$. It follows by Lemma~\ref{l:multi-decomp} that $f_0$ is a \bimorph{}
    \[ \Pc \to \left[\vec{\Pc_i}\right]. \]
    Therefore, by Proposition~\ref{p:bimorphisms}, there exists a coalgebra morphism $e_0\colon \pi \to \lop\left(\vec{\pi_i}\right)$ such that $e_0^\# = f_0$ and, furthermore, by Proposition~\ref{p:bimorph-dist}
    \[
        e^\# = f = \op\left(\veci e\right)\circ f_0 = \op\left(\veci e\right)\circ e_0^\# = (\lop\left(\veci e\right) \circ e_0)^\#
    \]
    which proves that $e = \lop\left(\veci e\right) \circ e_0$, since $(-)^\#$ is a bijection between coalgebra morphisms $(P, \pi) \to \lop\left(\vec{\alpha_i}\right)$ and \bimorph{}s $\pi \to \left[\vec{\alpha_i}\right]$.

    Lastly, we show the minimality of this decomposition. If $e$ decomposes as $g_0\colon \Pc \to \lop\left(\vec{\Qc_i}\right)$ followed by $\lop\left(\veci g\right)$, with embeddings $g_i\colon (\Qs_i,\Qc_i) \embed (\As_i,\Ac_i)$, for $1 \leq i \leq n$, then by Proposition~\ref{p:bimorph-dist},
    \[
        f = e^\# = (\lop\left(\veci g\right) \circ g_0)^\# = \op\left(\veci g\right) \circ g_0^\#
    \]
    Therefore, by minimality of the decomposition $f$ as $\op\left(\veci e\right) \circ f_0$, there exist coalgebra morphisms $h_i \colon (\Ps_i, \Pc_i) \to (\Qs_i, \Qc_i)$ such that $e_i = g_i \circ h_i$, for every $i\in \{1,\dots,n\}$.
\end{prf}

\begin{rem}
In fact, \ref{ax:s2} is equivalent to the version of \ref{ax:s2p} where we only care about morphisms $\Ps \to \op(\vec{\As_i})$ which arise from embeddings, via the correspondence in Proposition~\ref{p:bimorphisms}.
\end{rem}

Finally, by putting together the results from the last two sections, we are ready to state and prove an FVM theorem for the full logic. This time we state all assumptions for ease of reference.

\begin{restatable}{thm}{FVMfulllogic}
    \label{t:fvm-standard}
    Let $\op\colon \prod\nolimits_i \CC_i \to \CD $ be a functor and let $\C_1,\dots,\C_n$ and $\D$ be comonads on categories $\CC_1,\dots,\CC_n$ and $\CD$, respectively, such that each of the arising categories of coalgebras are path categories and, furthermore, \(\EM\D\) has equalisers.

    If \(\op\) preserves embeddings and there exists a Kleisli law of type
    \(\D\circ H \Rightarrow H \circ \prod\nolimits_i\C_i \)
    satisfying \ref{ax:s2p}, then
    \[ A_1 \lequiv{\C_1} B_1, \;\ldots,\; A_n \lequiv{\C_n} B_n \]
    implies
    \[ H(A_1,\dots,A_n) \;\lequiv{\D}\; H(B_1,\dots,B_n). \]
\end{restatable}
\begin{prf}
    Given a collection of open pathwise-embeddings
    \begin{align*}
        \EMF{\C_1}(A_1) \xleftarrow{f_1} Z_1 \xrightarrow{g_1} \EMF{\C_1}(B_1),
        \quad\dots\ , \ \
        \EMF{\C_n}(A_n) \xleftarrow{f_n} Z_n \xrightarrow{g_n} \EMF{\C_n}(B_n)
    \end{align*}
    in $\EM{\C_1},\,\dots,\,\EM{\C_n}$, respectively, we wish to construct a pair of open pathwise-embeddings
    \begin{align}
        \EMF{\D}(\op(A_1,\dots,A_n)) \leftarrow Z \rightarrow \EMF{\D}(\op(B_1,\dots,B_n))
        \label{eq:free-op-cospan}
    \end{align}
    in $\EM \D$.

    Since there is a Kleisli law $\kappa \colon \D\circ \op \Rightarrow \op \circ \prod_i\C_i$, we know by Theorem~\ref{t:em-lifting} that $H$ lifts to $\lop \colon \prod_i \EM{\C_i} \to \EM \D$.
    Also, since $\op$ preserves embeddings, we know that \ref{ax:s1} holds by Proposition~\ref{p:s1p}.
    Furthermore, \ref{ax:s2} holds too by Proposition~\ref{p:s2p} because we assumed that $\kappa$ satisfies \ref{ax:s2p}.
    Consequently, by Theorem~\ref{t:smoothness}, $\lop$ sends tuples of open pathwise-embeddings to open pathwise-embeddings.

    We obtain a pair of open pathwise-embeddings
    \[
    \begin{tikzcd}[column sep=3.5em]
        \lop(\vec{\EMF{\C_i}(A_i)})
        & \lop(\veci Z)
        \ar[swap]{l}{\lop(\veci f)}
        \ar{r}{\lop(\veci g)}
        &\lop(\vec{\EMF{\C_i}(B_i)})
    \end{tikzcd}
    \]
    as the morphisms of coalgebras $f_1,\dots,f_n$ and $g_1,\dots,g_n$ are open pathwise-embeddings too.

    Finally, we recall that, by Theorem~\ref{t:em-lifting},
    \[
        \EMF{\D}(\op(\veci A))
        \cong
        \lop(\vec{\EMF{\C_i}(A_i)}))
        \qtq{and}
        \EMF{\D}(\op(\veci B))
        \cong
        \lop(\vec{\EMF{\C_i}(B_i)})).
    \]
    This gives us that $\lop(f_1,\dots,f_n)$ and $\lop(g_1,\dots,g_n)$ are of the required type, as in \eqref{eq:free-op-cospan}.
\end{prf}

\begin{rem}
Recall that the requirement that \(\EM\D\) has equalisers in Theorem~\ref{t:fvm-standard} holds automatically in our examples, by Corollary~\ref{c:equalisers}.
    Furthermore, when we say that a Kleisli law \(\kappa\colon \D\circ H \Rightarrow H \circ \prod\nolimits_i\C_i \) satisfies \ref{ax:s2p}, although \(\kappa\) is never mentioned in the phrasing of the axiom \ref{ax:s2p}, it appears implicitly in the definition of \bimorph{}s.
\end{rem}

\begin{exa}
  \label{e:coproducts-full}
  We can apply Theorem~\ref{t:fvm-standard} to show that $\sequiv{\Pk}$ is preserved by taking disjoint unions of structures.
  As the functor~$\uplus$ preserves embeddings, we only have to check~\ref{ax:s2p} for the Kleisli law~$\kappa_{A_1,A_2}\colon \Pk(A_1 \uplus A_2) \to \Pk(A_1) \uplus \Pk(A_2)$ defined in Example~\ref{ex:coproducts-pe}.
  For a path $(\Ps,\Pc)$ in $\EM\Pk$, consider a morphism $f\colon \Ps \rightarrow A_1 \uplus A_2$ , where $(A_1,\Ac_1)$ and $(A_2,\Ac_2)$ are $\Pk$-coalgebras, and $f$ makes the following diagram of \ref{ax:s2p} commute.
    \[
      \begin{tikzcd}
          \Ps
            \ar{rr}{f}
            \ar[swap]{d}{\Pc}
          &
          &
          \As_1 \uplus \As_2
            \ar{d}{\alpha_1 \uplus \alpha_2}
          \\
          \Pk(\Ps)
            \ar{r}{\Pk(f)}
          &
          \Pk(\Ac_1 \uplus \Ac_2)
            \ar{r}{\kappa}
          &
          \Pk(\As_1) \uplus \Pk(\As_2)
      \end{tikzcd}
    \]
    Recall from Remark~\ref{r:coalg-order} that the coalgebra map $\Pc$ defines a forest order $\sqsubseteq_\Pc$ on the universe of $\Ps$, and similarly for coalgebra maps $\alpha_1,\alpha_2$. Furthermore, since $(\Ps,\Pc)$ is a path, $(P,\sqsubseteq_\Pc)$ is a finite linear order $x_1 \sqsubseteq_\Pc \dots \sqsubseteq_\Pc x_n$.

    We may assume without loss of generality that $f(x_n)$ is in $\As_1$. Then, $(\alpha_1 \uplus \alpha_2)(f(x_i))$ gives a word in $\Pk(\Ac_1)$, which forms a finite linear order in the $\sqsubseteq_{\alpha_i}$ order. By commutativity of the above diagram, this must be the same word as the down-to-right composition $\kappa(\Pk(f)(\Pc(x_n)))$. But, by definition, the latter is just the reduction of the word on elements $f(x_1), \dots, f(x_n)$ to the positions where $f(x_i) \in A_1$.

    The same reasoning applies for the largest $j$ such that $f(x_j)$ is in $\As_2$.
    Each of the two words obtained this way yields a path embedding $e_i\colon (P_i,\pi_i) \emb (A_i,\alpha_i)$ as the embedding of the induced substructure of $A_i$ on given word letters.
    Define a morphism $e_0\colon P \rightarrow P_1 \uplus P_2$ by sending an element to its image under $f$. By construction, $(e_1 \uplus e_2) \circ e_0$ forms a decomposition of $f$ and minimality of this decomposition is immediate, as any other decomposition of $f$ would contain $P_i$ as a subpath. We obtain the following FVM theorem for the $k$ variable fragment:
    \begin{align*}
    &A_1 \lequiv{\Pk} B_1 \ete{and}  A_2 \lequiv{\Pk} B_2 \\
    &\tq{implies}
    A_1 \uplus A_2 \lequiv{\Pk} B_1 \uplus B_2
    \end{align*}

The same reasoning goes through for $\Pk$ and $\Ek$ with coproducts over any index set $I$ and for $\Mk$ with $\merge{R}$ from Example~\ref{ex:coproducts-with-choice-pe}.
\end{exa}

\section{FVM theorems for existential fragments}
\label{s:existential}
In \cite{ALRXXexistential} it is observed that we can also express preservation of existential formulas naturally in the setting of game comonads. Recall that \df{existential formulas} are first-order formulas without universal quantifiers and with negation only in front of atoms. Define
\begin{itemize}
    \item $A \earrow{\C} B$\, if there exists a pathwise-embedding $\EMF\C(A) \to \EMF\C(B)$.
\end{itemize}
Then, it follows that \(A \earrow{\Ek} B\) iff, for every existential sentence $\varphi$ of quantifier depth at most $k$, \(A \models \varphi\) implies \(B \models \varphi\).
A similar statement holds for \(A \earrow{\Pk} B\) and the \(k\)-variable restriction of the existential fragment \cite{abramskyreggio2023arboreal}.

It turns out that, by Lemma~\ref{l:pe-preserved}, under the same conditions as in Theorem~\ref{t:fvm-standard} we have a preservation theorem for \(\earrow{\C}\). In fact, a weaker version of axiom \ref{ax:s2p} stated in terms of \bimorph{}s suffices. This weaker axiom no longer requires the decomposition to be minimal, that is, we define:
\begin{axioms}
        \item[\em\textbf{\namedlabel{ax:s2pp}{(S2'')}}]
        For any path $(\Ps,\Pc)$ in $\EM \D$ and coalgebras $(\As_i,\Ac_i)$ in $\EM{\C_i}$ for $1 \leq i \leq n$,
        every \bimorph{} $\Pc \to [\veci \Ac]$
        has \emph{some} decomposition through
        \[ \op(\veci e)\colon \op(\vec{\Ps_i}) \to \op(\vec{\As_i}),\]
        for some path embeddings $e_i\colon (\Ps_i,\Pc_i) \embed (\As_i,\Ac_i)$.
\end{axioms}
The corresponding FVM theorem for the existential fragments reads as follows.

\begin{thm}
    Let $\op\colon \prod\nolimits_i \CC_i \to \CD $, $\C_1,\dots,\C_n$ and $\D$ be exactly as in Theorem~\ref{t:fvm-standard}, except that instead of \ref{ax:s2p} we only require \ref{ax:s2pp} for the Kleisli law \(\D\circ H \Rightarrow H \circ \prod\nolimits_i\C_i \). Then
    \[ A_1 \earrow{\C_1} B_1, \;\ldots,\; A_n \earrow{\C_n} B_n \]
    implies
    \[ H(A_1,\dots,A_n) \;\earrow{\D}\; H(B_1,\dots,B_n). \]
\end{thm}
\begin{proof}
    Observe that in the proof of Lemma~\ref{l:pe-preserved} the minimality of the decomposition was not needed. Further observe that the variant of \ref{ax:s2} where minimality is not required is obtained from \ref{ax:s2pp} by the same proof as for Proposition~\ref{p:s2p}.
    Consequently, we see that the lifting $\lop\colon \prod\nolimits_i \EM{\C_i} \to \EM \D$ of \(\op\) preserves pathwise-embeddings.

    Then, from \(A_1 \earrow{\C_1} B_1, \;\ldots,\; A_n \earrow{\C_n} B_n\), witnessed by pathwise-embeddings
    \[
        f_1\colon \EMF{\C_1}(A_1) \to \EMF{\C_2}(B_2),
        \;\ldots,\;
        f_n\colon \EMF{\C_n}(A_n) \to \EMF{\C_n}(B_n)
    \]
    we obtain a pathwise-embedding
    \[\lop(f_1,\ldots,f_n)\colon \lop(\EMF{\C_1}(A_1), \ldots, \EMF{\C_n}(A_n)) \to \lop(\EMF{\C_1}(B_1), \ldots, \EMF{\C_n}(B_n)).\]
    Finally, since \(\lop\) commutes with the cofree functors, i.e.\ \(\lop(\EMF{\C_1}(A_1), \ldots, \EMF{\C_n}(A_n)) = \EMF\D(H(A_1,\dots,A_n))\), we have the required relation \(H(A_1,\dots,A_n) \earrow{\D} H(B_1,\dots,B_n)\).
\end{proof}

\begin{rem}
    The recent papers of Abramsky, Laure and Reggio \cite{abramskyreggio2024invitation,abramskylaurereggio2025existential+positive} detail the categorical description of the relation \(\larrow{{}^+\C}\) corresponding to the \emph{positive fragments}, i.e.\ fragments without negation.
    For two structures \(A,B\), we write \(A\larrow{{}^+\C} B\) if there exists a \emph{positive bisimulation} from \(F^\C(A)\) to \(F^\C(B)\), that is, there is a diagram of the form
    \[
        \begin{tikzcd}
            Z_1 \dar[swap]{f}\rar{h} & Z_2 \dar{g}
            \\
            F^\C(A) & F^\C(B)
        \end{tikzcd}
    \]
    where \(f,g\) are open pathwise-embeddings and \(h\) is a \emph{bijection}.
    Importantly, the definition of bijections requires \(\EM\C\) to be arboreal.

    Consequently, in order to describe FVM theorems for the \(\larrow{{}^+\C}\) relation, we need to describe conditions under which functors between arboreal categories of coalgebras preserve bijections.
    We leave this as an open problem for future investigation.
\end{rem}

\section{Abstract FVM theorems for products}
\label{sec:product-theorems}
In this section we show that FVM theorems for the operation of the categorical product of two structures are automatic, regardless of the chosen comonad. This shows the power of the categorical approach since the theorem applies to any situation where logical equivalence admits a comonadic characterisation.

Recall that the (categorical) product $\prod_i \As_i$
in the category \(\Rel\),
indexed by a set $I$, is the \(\sg\)-structure with the universe being the product of the underlying universes $\prod_i A_i$. Further, for a relation $R$ in $\sg$, define its realisation in the product by
\[
    R^{\prod_i A_i}(\ol a_1,\dots, \ol a_n) \iff \ (\forall i) \ \ R^{A_i}(a_{1,i},\, \dots, \, a_{n,i})
\]
This operation obeys the usual universal property of categorical products.
Namely, for any object $C$ and morphisms $h_i\colon C\to A_i$ for each \(i \in I\), there is a \emph{unique} morphism $\overline h\colon C\to \prod_i \As_i$ such that,
for every \(i\in I\),
\[
    \begin{tikzcd}
        C \rar{\ol h} \ar[swap]{rd}{h_i} & \prod_i A_i \dar{\pi_i} \\
        & A_i
    \end{tikzcd}
\]
where $\pi_i\colon \prod_i \As_i \to A_i$ is the $i$th projection. In case of products of relational structures, $\overline h$ sends $c$ to the tuple~$(h_i(c))_{i\in I}$.
Products in $\Rels$ work similarly, with the distinguished element of $\prod_i (\As_i,a_i)$ being the tuple~$(a_i)_i$.

\subsection{The positive existential and counting fragments}
We fix an arbitrary comonad~$\C$ on a category~\(\CC\) with products.
Observe that, for a family of objects $\{A_i\}_i$ of $\CC$, indexed by a set \(I\), we have morphisms
\(
    \C(\pi_i) \colon \C(\prod\nolimits_i A_i) \to \C(A_i)
\)
for each $i \in I$ and, by the universal property of products, also the morphism
\[\kappa_{\veci A} \colon \C(\prod\nolimits_i A_i) \to \prod\nolimits_i \C(A_i). \]

Applying Theorem~\ref{t:fvm-pe} immediately yields the following FVM theorem for positive existential fragments.
\begin{thm}
    \label{t:product-fvm-pe-categorical}
    Given collections of $\sg$-structures $\{A_i\}_{i\in I}$ and $\{B_i\}_{i\in I}$, indexed by a common set $I$, and a comonad $\C$ on a category $\CC$ with products. Then,
    \[
        \textstyle
        \left( \forall i\in I.\; A_i \parrow{\C_i} B_i \right) \qtq{implies} \prod_i A_i \parrow{\C_i} \prod_i B_i.
    \]
\end{thm}

A direct calculation reveals that $\kappa$ is natural in the choice of $\{A_i\}_i$ and, furthermore, is a Kleisli law.

\begin{lem}
    \label{l:kappa-products}
    $\kappa$ is a Kleisli law.
\end{lem}
\begin{prf}
    First, observe that by the definition of products and how $\kappa$ is defined, we have that
    \begin{equation}
    \label{eq:prod-prop}
    \C(\pi_i) = \pi_i \circ \kappa
    \end{equation}
    for every $i$. From this we show commutativity of the following diagram.
    \[
        \begin{tikzcd}
            \C\left(\prod_i A_i\right) \ar{rr}{\kappa} \ar[swap]{d}{\counit} & & \prod_i \C(A_i) \ar{d}{\prod_i \counit} \\
            \prod_i A_i \ar[swap]{dr}{\pi_i} & & \prod_i A_i \ar{dl}{\pi_i} \\
            & A_i
        \end{tikzcd}
    \]
    By universality of the products, this implies the counit axiom of Kleisli laws that is \ref{ax:kl-law-counit} for the operation $\op(\veci A) = \prod_i A_i$. To show commutativity of the pentagon above, compute:
    \begin{align*}
        &\pi_i \circ \left(\prod\nolimits_i \counit\right) \circ \kappa \\
        &= \counit \circ \pi_i \circ \kappa & \text{products} \\
        &= \counit \circ \C(\pi_i) & \eqref{eq:prod-prop} \\
        &= \pi_i \circ \counit & \text{naturality}
    \end{align*}
    Next, to show \ref{ax:kl-law-comultiplication} for the operation $\op(\veci A) = \prod_i A_i$, we employ the same strategy. We wish to show commutativity of the following diagram.
    \[
        \begin{tikzcd}[column sep=5em]
            \C(\prod_i A_i) \ar{r}{\kappa} \ar[swap]{d}{\delta} & \prod_i \C(A_i) \ar{dd}{\prod_i \delta} \\
            \C^2(\prod_i A_i) \ar[swap]{d}{\C(\kappa)} \\
            \C(\prod_i \C(A_i)) \ar[swap]{d}{\kappa} & \prod_i \C^2(A_i) \ar{d}{\pi_i} \\
            \prod_i \C^2(A_i) \ar{r}{\pi_i} & \C^2(A_i)
        \end{tikzcd}
    \]
    To do so, we compute:
    \begin{align*}
        &\pi_i \circ \left(\prod\nolimits_i \delta\right) \circ \kappa  \\
        &= \delta \circ \pi_i \circ \kappa & \text{products} \\
        &= \delta \circ \C(\pi_i) & \eqref{eq:prod-prop} \\
        &= \C^2(\pi_i) \circ \delta & \text{naturality} \\
        &= \C(\pi_i) \circ \C(\kappa) \circ \delta  & \eqref{eq:prod-prop} \\
        &= \pi_i \circ \kappa \circ \C(\kappa) \circ \delta  & \eqref{eq:prod-prop} & \qedhere
    \end{align*}
\end{prf}

As a corollary of Lemma~\ref{l:kappa-products} and Theorem~\ref{t:fvm-counting}, we obtain an abstract FVM for counting fragments.
\begin{thm}
    Given the same assumptions as in Theorem~\ref{t:product-fvm-pe-categorical},
    \[
        \textstyle
        \left( \forall i\in I.\; A_i \cequiv{\C} B_i \right)
        \qtq{implies}
        \prod_i A_i \cequiv{\C} \prod_i B_i.
    \]
\end{thm}

\subsection{The full fragments}
We also have an FVM theorem for products and the equivalence~$\lequiv \C$.
We prove this by checking the assumptions of Theorem~\ref{t:fvm-standard}.
We first check that products preserve embeddings.
\begin{lem}
    \label{l:prod-s1p}
    Let \(\CC\) be a category with products and a weak factorisation system \((\Qu,\Em)\). Then, for any collection of embeddings \(\{m_i\colon C_i \emb D_i\}_i\) the product \(\prod_i m_i\colon \prod_i C_i \to \prod_i D_i\) is an embedding too.
\end{lem}
\begin{proof}
    From the definition of weak factorisation systems it is enough to show that \(\prod_i m_i\) is weakly orthogonal to any quotient \(e\in \Qu\), i.e.\ \(e \pitchfork \prod_i m_i\). To this end, assume we have the following a commutative square.
    \[
        \begin{tikzcd}[column sep=3.5em]
            A \rar[->>]{e}\dar[swap]{f} & B \dar{g}\\
            \prod_i C_i \rar{\prod_i m_i} & \prod_i D_i
        \end{tikzcd}
    \]
    Upon fixing \(i\), we can compose this square with the naturality square for product projections.
    \[
        \begin{tikzcd}[column sep=3.5em]
            \prod_i C_i \dar[swap]{\pi_i}\rar{\prod_i m_i} & \prod_i D_i \dar{\pi_i} \\
            C_i \rar[>->]{m_i} & D_i
        \end{tikzcd}
    \]
    Then, since \(e \pitchfork m_i\), there is some \(d_i\colon B\to C_i\) such that \(d_i \circ e = \pi_i \circ f\) and \(m_i \circ d = \pi_i \circ g\).

    Set \(\langle\veci d\rangle\) to be the unique morphism \(B\to \prod_i C_i\) obtained from the cone \(\{d_i\colon B\to C_i\}_i\) and the universal property of products. Observe that
    \[
        \pi_i \circ \langle\veci d\rangle \circ e = d_i \circ q = \pi_i \circ f
    \]
    which gives \(\pi_i \circ \langle\veci d\rangle = f\) and, similarly,
    \[
        \pi_i \,\circ\, \prod\nolimits_i m_i \,\circ\, \langle\veci d\rangle
        = m_i \circ \pi_i \circ \langle\veci d\rangle
        = m_i \circ d_i
        = \pi_i \circ g
    \]
    which gives the second required equation, i.e.\ \(\prod_i m_i \circ \langle\veci d\rangle = g\).
\end{proof}

Recall that in order to specify $\lequiv \C$, we need \(\EM\C\) to be a path category, that is, it comes equipped with a proper factorisation system \((\ol \Qu, \ol \Em)\) and a choice of path objects $\Pa \sue\EM \C$.
As before, we consider the setting where \((\ol \Qu, \ol \Em)\) is obtained as the lifting of a proper factorisation system \((\Qu,\Em)\) on the base category \(\CC\), as in Lemma~\ref{l:lift-fs}.
In this case, a morphism of coalgebras $(A,\alpha) \to (B,\beta)$ is a quotient (i.e.\ is in \(\ol\Qu\)) if the underlying morphism $A \to B$ is a quotient (i.e.\ is in \(\Qu\)).

The crucial requirement of our abstract FVM theorems for products is that \df{paths are closed under quotients}, that is, for any quotient
\[ X \quot Y \]
if $X$ is a path then so is $Y$. The following lemma shows that this is a sufficient requirement for our path category \(\EM\C\) to ensure the second assumption of Theorem~\ref{t:fvm-standard}.

\begin{lem}
    \label{l:prod-s2p}
    If $\C$ preserves embeddings and paths are closed under quotients in \(\EM\C\),
    then $\op(\veci A) = \prod_i A_i$ satisfies axiom \ref{ax:s2p}.
\end{lem}
\begin{prf}
    We assume that $(Q,\rho)$ is a path in $\EM \C$ and $(A_i,\alpha_i)$ are arbitrary $\C$-coalgebras. Further, assume that we are given a morphism $f\colon \Ps \to \prod_i A_i$ which makes the following diagram commute.
    \[
    \begin{tikzcd}[->]
        \Qs \arrow[rr, "f"] \dar[swap]{\Qc} & & \prod_i\As_i \dar{\prod_i \Ac_i} \\
        \C(\Qs) \rar{\C(f)} & \C(\prod_i \As_i) \rar{\kappa} & \prod_i \C(\As_i)
    \end{tikzcd}
    \]
    It is an easy observation that the composition $\pi_i \circ f$ with the \(i\)th projection \(\pi_i \colon \prod_i A_i \to A_i\) is a coalgebra morphism $(Q,\rho) \to (A_i, \alpha_i)$. Since \(\C\) preserves embeddings, Lemma~\ref{l:lift-fs} applies and we can take the factorisation of this morphism in \(\EM\C\) as
    \[
        \begin{tikzcd}
            (Q,\rho) \rar[->>]{q_i} & (Q_i, \rho_i) \rar[>->]{e_i} & (A_i, \alpha_i)
        \end{tikzcd}
    \]
    where \(q_i\) is a quotient and $e_i$ is an embedding. Since $(Q,\rho)$ is a path and $q_i$ is a quotient, we see that~$(Q_i, \rho_i)$ is a path by our assumptions. Consequently, we have a decomposition of $f$ into
    \begin{equation}
        \label{eq:f-decomp-prod-1}
        \begin{tikzcd}[column sep=3.0em]
            Q \rar{q} & \prod_i Q_i \rar{\prod_i e_i} & \prod_i A_i
        \end{tikzcd}
    \end{equation}
    where $q$ is obtained from the universal property of products, from the collection of morphisms~$\{q_i\}_i$. %

    To show the minimality of this decomposition, let $f$ also decompose as
    \begin{equation}
        \label{eq:f-decomp-prod-2}
        \begin{tikzcd}[column sep=3.0em]
            Q \rar{r} & \prod_i S_i \rar{\prod_i m_i} & \prod_i A_i
        \end{tikzcd}
    \end{equation}
    where $m_i$ are path embeddings $(S_i, \omega_i) \emb (A_i, \alpha_i)$ in $\EM \C$. By post-composing both homomorphisms in \eqref{eq:f-decomp-prod-1} and \eqref{eq:f-decomp-prod-2} with the projection $\pi_i\colon \prod_i A_i \to A_i$ we obtain a commutative square:
    \[
        \begin{tikzcd}
            Q \ar{r}{q_i} \ar[swap]{d}{\pi_i \circ r} & Q_i \ar{d}{e_i} \\
            S_i \ar{r}{m_i} & A_i
        \end{tikzcd}
    \]
    Now, observe that $q_i$ is a quotient and $m_i$ is an embedding and so there is a diagonal morphism $d_i\colon Q_i \to S_i$ such that $e_i = m_i \circ d_i$. It remains to check that $d_i$ is a morphism of coalgebras $(Q_i, \rho_i) \to (S_i, \omega_i)$. This is a consequence of the fact that $\C(m_i)$ is an embedding and that the right square and the surrounding rectangle of the following diagram commute.
    \[
        \begin{tikzcd}
            Q_i \ar{r}{d_i} \ar[swap]{d}{\rho_i} & S_i \ar{r}{m_i} \ar{d}{\omega_i} & A_i \ar{d}{\alpha_i} \\
            \C(Q_i) \ar{r}{\C(d_i)} & \C(S_i) \ar{r}{\C(m_i)} & \C(A_i)
        \end{tikzcd}
    \]
    Indeed, let $x\in Q_i$. By chasing the above diagram we obtain
    \begin{align*}
        \C(m_i)\circ \omega_i \circ d_i
        = \alpha_i \circ m_i \circ d_i
        = \C(m_i) \circ \C(d_i) \circ \rho_i
    \end{align*}
    and since $\C(m_i)$ is a monomorphism $\omega_i \circ d_i = \C(d_i)\circ \rho_i$, as required.
\end{prf}

To summarise, we have shown that any category \(\CC\) with products admits a Kleisli law for any comonad \(\C\) on \(\CC\) (Lemma~\ref{l:kappa-products}). If, furthermore, \(\EM\C\) has equalisers then \(\prod_i\) automatically satisfies assumptions of Theorem~\ref{t:fvm-standard} by Lemmas~\ref{l:prod-s1p} and \ref{l:prod-s2p}.

As a corollary we obtain a general FVM theorem for products, with assumptions that hold true for all the game comonads studied so far in the literature (cf.\ Example~\ref{ex:arbo-fvm-product} below).

\begin{thm}
    Let \(\C\) be a comonad on a category \(\CC\) with products, such that \(\EM\C\) is a path category. If \(\EM\C\) has equalisers and paths are closed under quotients then
    \[
        \textstyle
        \left(\forall i\in I.\; A_i \lequiv{\C} B_i \right)
        \qtq{implies}
        \prod_i A_i \,\lequiv{\C}\, \prod_i B_i.
    \]
\end{thm}

As a consequence of Corollary~\ref{c:equalisers} we obtain the following specialised version of the above theorem.

\begin{cor}
    \label{cor:abstract-fvm-products-full}
    If a comonad $\C$ on \(\Rel\) or \(\Rels\) preserves embeddings and if \(\EM\C\) is a path category with paths closed under quotients, then
    \[
        \textstyle
        \left(\forall i\in I.\; A_i \lequiv{\C} B_i \right)
        \qtq{implies}
        \prod_i A_i \,\lequiv{\C}\, \prod_i B_i.
    \]
\end{cor}

\begin{exa}
    \label{ex:ek-pk-mk-products}
In the case of $\Ek$,
it is immediate that $\Ek$ preserves embeddings. Furthermore, for a surjective coalgebra morphism $(\Ps,\Pc) \to (\As,\Ac)$ in $\EM \Ek$ such that $(\Ps,\Pc)$ is a path, we also have that $(\As,\Ac)$ is a path since, in view of Remark~\ref{r:coalg-order}, we map a finite linear order onto a forest.

The same is true for our example comonads $\Pk$, $\Mk$, yielding product FVM theorems for the bounded quantifier rank, bounded variable and bounded modal depth fragments and their positive existential and counting variants.
\end{exa}

\begin{exa}
    \label{ex:arbo-fvm-product}
More generally, the assumptions of Theorem~\ref{cor:abstract-fvm-products-full} hold for any comonad $\C$ which preserves embeddings and such that the category $\EM \C$ is an \df{arboreal category}, in sense of \cite{AbramskyR21}, by Lemma 3.5 therein.

In particular, all FVM theorems for products from this section hold for \df{hybrid logic} because of the hybrid comonad of~\cite{AbramskyMarsden2022}, \df{guarded logics} captured by guarded comonads in~\cite{AbramskyM21}, the \df{bounded conjunction} and \df{bounded quantifier rank} fragments of the $k$-variable first-order logic captured in \cite{montacute2021pebble} and \cite{Paine20}, respectively.
\end{exa}

\begin{rem}
Since monadic second-order logic (MSO) does not have a FVM theorem for products (see e.g.~\cite{makowsky2004algorithmic}), we can conclude that there does not exist a comonad for MSO over \(\Rel\) which satisfies the conditions of Corollary~\ref{cor:abstract-fvm-products-full}. 
This limitation can be overcome by encoding the standard semantics of MSO  using the technique of enrichment, described in Section~\ref{s:enrichments}, applied to a modified $\Ek$ comonad on two-sorted structures. 
The full details of this construction appear in the follow up paper~\cite{jaklmarsdenshah2022fvm}.
\end{rem}

\section{Adding equality and other enrichments}
\label{s:enrichments}

Recall that the relations $\sequiv{\Ek}$ and $\sequiv{\Pk}$ express logical equivalence with respect to the logic \emph{without} equality.
The standard way to add equality to the language is to extend the signature $\sg$ with an additional binary relation symbol $I(\cdot,\cdot)$, and provide a \df{translation}
$\trI \colon \Rel \to \R(\sg\cup \{I\})$
which interprets this new relational symbol as equality, i.e. $\trI(A)$ is the $\sg\cup \{I\}$-structure extension of $A$ by setting $I(a,b)$ if and only if $a = b$. Since our game comonads are defined uniformly for any signature, they are also defined over $\R(\sg\cup \{I\})$, giving us:

\begin{itemize}
    \item $\trI(A) \sequiv \Ek \trI(B)$ iff $A$ and $B$ are logically equivalent w.r.t.\ the fragment of first-order logic consisting of formulas of quantifier depth ${\leq}\, k$ with equality, and
    \item $\trI(A) \sequiv \Pk \trI(B)$ iff $A$ and $B$ are logically equivalent w.r.t.\ the fragment of first-order logic consisting of formulas in $k$-variables with equality.
\end{itemize}

Moreover, an operation which commutes with the translation $\trI$ lifts FVM theorems for fragments without equality to fragments with equality. For example, disjoint unions satisfy that $\trI(A \uplus B) = \trI(A) \uplus \trI(B)$ and, therefore, if $\trI(A_1) \sequiv \Ek \trI(B_1)$, and $\trI(A_2) \sequiv \Ek \trI(B_2)$, then
\begin{align*}
    \trI(A_1 \uplus A_2)  &\enspace=\enspace\mkern8mu \trI(A_1) \uplus \trI(A_2) \\
    &\enspace\sequiv{\Ek}
    \trI(B_1) \uplus \trI(B_2) = \trI(B_1 \uplus B_2)
\end{align*}
which shows that $A_1 \uplus A_2$ and $B_1 \uplus B_2$ are logically equivalent w.r.t.\ formulas of quantifier depth ${\leq}\, k$ with equality.

In general, similar translations can be used for other logic extensions.
For example, in \cite{BednarczykUrbanczyk2022comonadicdescriptionlogics} a translation is used to obtain description logics using $\Mk$.
Nevertheless, the above technique applies verbatim to varying translations, operations and logics.
We arrive at the following simple but important fact, that greatly increases the applicability of the results in the previous sections.
\begin{restatable}{thm}{CommutativityWithTranslations}
\label{thm:translations}
    For $1 \leq i \leq n+1$, let $\Lo_i$ and $\Lo'_i$ be logics and $\tr_i$ a translation such that ${A \lequiv{\Lo_i} B}$ \,iff\, ${\tr_i(A) \lequiv{\Lo'_i} \tr_i(B)}$. Further, let $\op'$ be an $n$-ary operation satisfying the \emph{FVM theorem from $\Lo'_1,\dots,\Lo'_n$ to $\Lo'_{n+1}$}, with respect to translated structures. That is:
    \[ \tr_1(A_1) \lequiv{\Lo'_1} \tr_1(B_1),\;\ldots,\;\tr_n(A_n) \lequiv{\Lo'_n} \tr_n(B_n)  \]
    implies
    \[ \op'(\vec{\tr_i(A_i)}) \lequiv{\Lo'_{n+1}} \op'(\vec{\tr_i(B_i)}) \]
    If $\op$ commutes with the translations in the sense that
    $\op'(\vec{\tr_i(A_i)}) \cong \tr_{n+1}(\op(\veci A))$,
    then $\op$ satisfies the FVM theorem from $\Lo_1,\dots,\Lo_n$ to $\Lo_{n+1}$. That is:
    \[
        A_1 \lequiv{\Lo_1} B_1, \;\ldots,\; A_n \lequiv{\Lo_n} B_n
        \qtq{implies}
        \op(\veci A) \lequiv{\Lo_{n+1}} \op(\veci B)
    \]
\end{restatable}
\begin{prf}
    Assume that $A_i \equiv_{\Lo_i} B_i$, for every $i=1,\dots,n$.
    By our assumptions, we obtain that $\tr_i(A_i) \equiv_{\Lo'_i} \tr_i(B_i)$ for every $i=1,\dots,n$ and, therefore,
    \[
        \op'(\tr_1(A_1),\dots,\tr_n(A_n)) \lequiv{\Lo'_{n+1}} \op'(\tr_1(B_1), \dots, \tr_n(B_n) )
    \]
    Then, by a straightforward calculation
    \begin{align*}
        \tr_{n+1}(\op(\veci A))
        \cong \op'(\vec{\tr_i(A_i)})
        \equiv_{\Lo'_{n+1}} \op'(\vec{\tr_i(B_i)})
        \cong \tr_{n+1}(\op(\veci B))
    \end{align*}
    which, by our assumptions, implies that
    \[ \op(\veci A) \equiv_{\Lo_{n+1}} \op(\veci B) .\qedhere\]
\end{prf}

\begin{exa}
Let $\trCon\colon \Rel \to \R(\sgCon)$ be a translation of $\sg$-structures into the extension $\sgCon$ of $\sg\cup\{I\}$ with an extra binary relation $\Con(\cdot,\cdot)$. The interpretation of $\Con$ in $\trCon(A)$ consists of all pairs~$(a,b)$ appearing in the same component of the Gaifman graph of $A$, that is, pairs~$(a,b)$ such that there is a path $a \leftrightsquigarrow x_1 \leftrightsquigarrow \dots \leftrightsquigarrow x_n \leftrightsquigarrow b$ where $x \leftrightsquigarrow y$ holds whenever $x$ and $y$ are in a common tuple in an $R^A$ of some~$R\in \sg$.

It is easy to see that $\trCon(A) \sequiv \Ek \trCon(B)$ expresses logical equivalence in the bounded quantifier depth fragment extended with the connectivity predicate, akin to~\cite{bojanczyk2021separatorlogic,schirsiebvigny2022connectivity}. Since $\trCon(A \uplus B) \cong \trCon(A) \uplus \trCon(B)$, we obtain, by Theorem~\ref{thm:translations} and Example~\ref{e:coproducts-full}, an FVM theorem for disjoint unions and the quantifier rank $k$ fragment of said logic.
While still being polynomial-time computable, the connectivity relation is not definable in bounded quantifier depth fragments of first-order logic,
and so this is a proper extension.
\end{exa}

\begin{exa}
    Continuing from Example~\ref{ex:ek-pk-mk-products}, consider $\tr\colon \Rels \to \Rs(\sg^G)$ where $\sg^G$ is the extension of the modal signature $\sg$ with the \df{global relation} $G(a,b)$, true in $\tr(A,a)$ for any pair of elements. Then, it is immediate that $\tr((A,a) \times (B,b)) \cong \tr(A,a) \times \tr(B,b)$, giving us by Theorem~\ref{thm:translations} an FVM theorem for products and modal logic with global modalities.
\end{exa}

\begin{exa}

As an example of Theorem~\ref{thm:translations} where the two operations differ, we consider the notion of \textit{weak (bi)simulation} where transitions along a distinguished relation $S \in \sg$ are considered to occur silently.
We can model weak bisimulation by considering bisimulation between translated structures where $\tr\colon \Rels \to \Rels$ replaces all relations $R$ with their closure under sequences of prior and subsequent $S$ transitions.

For pointed structures $\As_1 = (A_1,a_1)$ and $\As_2 = (A_2,a_2)$, there is an operation~$\As_1 \vee \As_2$ which adds a new initial point~$\star$ to $A_1 \uplus A_2$.  We extend the transitions with $R(\star, (i,x))$, for $R\in \sg$, whenever $R^{A_i}(a_i,x)$ for some $i\in \{1,2\}$ and $x \in A_i$.
This construction satisfies
$ \tr(\As_1 \merge{S} \As_2) \cong \tr(\As_1) \vee \tr(\As_2)$
by design, and there is a Kleisli law
\[\kappa_{\As_1,\As_2}\colon \Mk(\As_1 \vee \As_2)\rightarrow \Mk(\As_1) \vee \Mk(\As_2),\]
via a similar construction to that used in Example~\ref{ex:coproducts-with-choice-pe}. This  yields an FVM theorem:
\newcommand{\wksim}{\Rrightarrow_{W_k}}
\begin{align*}
&\As_1 \wksim \Bs_1 \ete{and} \As_2 \wksim \Bs_2 \\
&\tq{implies}
\As_1 \merge{S} \As_2 \wksim \Bs_1 \merge{S} \Bs_2
\end{align*}
where $\As \wksim \Bs$ indicates that $\As$ weakly simulates $\Bs$ up to depth~$k$.
\end{exa}

\section{Conclusion}

We presented a categorical approach to the composition methods, and specifically Feferman--Vaught--Mostowski theorems.
We exploit game comonads to encapsulate the logics and their model comparison games.
Surprising connections to classical constructions in category theory, and especially the monad theory of bilinear maps emerged from this approach, cf.\ Section~\ref{s:lifting-operations}. %

For finite model theorists, our work provides a novel high-level account of many FVM theorems, abstracting away from individual logics and constructions.
Furthermore, concrete instances of these theorems are verified purely semantically, by finding a suitable collection homomorphisms forming a Kleisli law satisfying \ref{ax:s2p}, instead of the usual delicate verification that strategies of model comparison games compose.
For game comonads, we provide a much needed tool that enables us to handle logical relationships between structures as they are transformed or viewed in terms of different logics.

The FVM results in this paper, combined with judicious use of the techniques described in Section~\ref{s:enrichments} can be pushed significantly further.
FVM theorems and the compositional method are of particular significance in the setting of monadic second-order (MSO) logic~\cite{gurevich1985monadic}. Exploiting the results we have presented for the composition methods, the follow up paper \cite{jaklmarsdenshah2022fvm} gives a comonadic semantics for MSO, and develops a semantic account of Courcelle's algorithmic meta-theorems~\cite{courcelle2012graph}.

The original theorem of Feferman--Vaught~\cite{feferman1959first} is very flexible, including incorporating structure on the indices of families of models to be combined by an operation, as well as the models themselves. This aspect is currently outside the scope of our comonadic methods.
An ad-hoc adaptation of our approach also gives a proof of the FVM theorem for free amalgamations, but this does not follow directly from our theorems. Developing these extensions is left to future work.

The present work bears some resemblance to Turi and Plotkin's bialgebraic semantics~\cite{turi1997towards}, a categorical model of structural operational semantics (SOS)~\cite{plotkin1981}.
Turi and Plotkin noticed that the SOS rules assigning behaviour to syntax could be abstracted as a certain distributive law $\lambda$.
An algebra encodes the composition operations of the syntax, and a coalgebra the behaviour.
If this pair is suitably compatible with~$\lambda$, forming a so-called $\lambda$-bialgebra, then crucially bisimulation is a congruence with respect to the composition operations.
Both bialgebraic semantics and our work presented in this paper give
a categorical account of well-behaved composition operations, with the interaction between composition and observable behaviour mediated by some form of distributive law. There are also essential differences: bialgebraic semantics typically focuses on assigning behaviour to syntax,
whereas FVM theorems encompass operations on models.
Our FVM theorems are parametric in a choice of logic or observable behaviour, while the notion of bisimulation in bialgebraic semantics is fixed by the coalgebra signature functor.
The natural presence of positive existential and counting quantifier variants of FVM theorems, and the incorporation of resource parameters, do not seem to have a direct analogue in bialgebraic semantics. The similarities are intriguing, and exploring the relationship between bialgebraic semantics and our approach to FVM theorems is left to future work.

\subsubsection*{Acknowledgements}
We would like to thank to the anonymous reviewers for their helpful and detailed comments, to Clemens Berger for suggesting the connection with parametric adjoints and to Nathanael Arkor for his suggestions on terminology.
We are also grateful to the members of the EPSRC project ``Resources and co-resources'' for their feedback.

\appendix

\counterwithin{thm}{section}
\counterwithin{prop}{section}
\counterwithin{lem}{section}
\counterwithin{defi}{section}

\section{Parametric relative right adjoints}
\label{s:parametric-relative-adjoints}

In this section we show how the axioms \ref{ax:s1} and~\ref{ax:s2} from Section~\ref{s:ope-preservation} relate to mathematical concepts that have previously appeared in the literature, specifically parametric relative adjoints.
Parametric adjoints, also sometimes called local adjoints, originate in the work of Street~\cite{street2000petit}.
We use their equivalent definition due to Weber~\cite{weber2004generic, weber2007familial,berger2012monads}.
Relative adjoints are a common tool in the semantics of programming languages, see for example \cite{altenkirch2015monads}, and are connected to relative comonads which can be used to encapsulate the translations discussed in Section~\ref{s:enrichments}~\cite{abramsky2021relating}. We will be interested in a combination of these two notions.

We first review the required terminology. By $\CT \slice A$ we denote the
\df{slice category}
consisting of pairs $(X,f)$, where $f\colon X\to A$ is a morphism in $\CT$. A morphism $(X,f) \to (Y,g)$ between objects in $\CT \slice A$ is a morphism $h\colon X\to Y$ such that $f = g \circ h$. A functor ${F\colon \CT \to \CS}$ is said to be a \df{parametric right adjoint} if, for each object $A$ of $\CC$, the functor
\[
    F_A : \CT \slice A \to \CS \slice F(A)
\]
sending $f\colon X\to A$ to $F(f)\colon F(X) \to F(A)$, is a right adjoint.
Finally, a functor $R\colon \CT \to \CU$ is a \df{relative right adjoint} between a category $\CT$ and a functor $I\colon \CS \to \CU$ if there is a functor $L\colon \CS \to \CT$ and a natural bijection between sets of morphisms
\begin{prooftree}
    \AxiomC{$L(S) \to T$ in $\CT$}
    \doubleLine
    \UnaryInfC{$I(S) \to R(T)$ in $\CU$}
\end{prooftree}

These definitions motivate our definition of parametric relative adjoints for categories of coalgebras with a selected class of path objects and a class of embeddings.
For simplicity we consider unary $\lop : \EM{\C} \to \EM{\D}$.
We write $\Pa_{\C}$ and $\Pa_{\D}$ for the class of paths in $\EM{\C}$ and $\EM{\D}$ respectively.

For $D\in \EM{\D}$, let $I_D$ be the inclusion functor
\[I_D\colon \Pa_{\D} \emcm D \to \EM{\D} \emcm D\]
where $\EM{\D} \emcm D$ and ${\Pa_{\D} \emcm D}$ denote the full subcategories of the comma category $\EM{\D} \slice D$ consisting of pairs~$(X,e)$ where $e$ is an embedding or a path embedding, respectively.

$\lop$ is a \df{parametric relative right adjoint} if, for every $C \in \EM{\C}$, the functor
    \[
        \lop_{C}\colon \Pa_{\C} \emcm C \ee\longrightarrow \EM{\D} \emcm \lop(C),
    \]
which sends a path embedding $e$, with $e\colon X \embed C$, to the embedding $\lop(e)$ of type  $\lop(X) \embed \lop(C)$, is a relative right adjoint between $\Pa_{\C} \emcm C$ and $I_{\lop(C)}$, which we will denote simply by~\(I\). Following the above terminology, the last condition assumes the existence of a functor $L\colon \Pa_{\D} \emcm \lop(C) \to \Pa_{\C} \emcm C$
such that there is a natural bijection between morphisms:
\begin{prooftree}
    \AxiomC{$L(e) \to f$ in $\Pa_{\C} \emcm C$}
    \doubleLine
    \UnaryInfC{$I(e) \to \lop_{C}(f)$ in $\EM{\D} \emcm \lop(C)$}
\end{prooftree}

By unwrapping the definitions, we see that the natural bijection sends an embedding \(m\colon P' \to Q\) which makes the diagram on the left below commute to an embedding \mbox{\(m^\dagger\colon P \to \lop(Q)\)} which makes the diagram on the right below commute.
\[
    \begin{tikzcd}[column sep=1.3em]
        P'
            \ar[>->]{rr}{m}
            \ar[swap,>->]{rd}{L(e)}
        && Q
            \ar[>->]{ld}{f}
        \\
        & C
    \end{tikzcd}
    \qquad
    \qquad
    \begin{tikzcd}[column sep=1.3em]
        P
            \ar[>->]{rr}{m^\dagger}
            \ar[swap,>->]{rd}{e}
        && \lop(Q)
            \ar[>->]{ld}{\lop(f)}
        \\
        & \lop(C)
    \end{tikzcd}
\]
Furthermore, the assignment \(m \mapsto m^\dagger\) has an inverse \(t \mapsto t_\dagger\) and the naturality can be rephrased by the following two rules:
\begin{itemize}
    \item For any \(g\colon Q \emb Q'\) and \(f' \colon Q' \emb C\) such that \(f = f' \circ g\):\quad \((g \circ m)^\dagger = \lop(g) \circ m^\dagger\).
    \item For any \(h\colon S \emb P\) and \(e' \colon S \to \lop(C)\) such that \(e' = e \circ h\):\quad \((t \circ h)_\dagger = t_\dagger \circ L(h)\).
\end{itemize}

In Proposition~\ref{p:prr-adj} below we show that under the assumption that \(\lop\) satisfies \ref{ax:s1}, the axiom \ref{ax:s2} is equivalent to \(\lop\) being a parametric relative right adjoint.
For the proof of the equivalence with parametric relative right adjoints we adapt a standard characterisation of left adjoints in terms of mapping on objects satisfying the expected universal property to the setting of relative adjoints. This is what we do next.

Recall that a functor $R\colon \CT \to \CS$ has a left adjoint if and only if, for every object $b\in \CS$, there is an object $L(b)$ in $\CT$ and a morphism $\eta_b\colon b \to R(L(b))$ such that, for any other $f\colon b \to R(a)$ there is a unique $f^*\colon L(b) \to a$ such that
\[
    \begin{tikzcd}
        b \rar{\eta_b}\ar[swap]{dr}{f} & R(L(b)) \dar{R(f^*)} \\
        & R(a)
    \end{tikzcd}
\]

A similar statement holds for relative adjoints too, with a proof that follows essentially the same steps.

\begin{lem}
    \label{l:rel-adj-univ}
    Let $R\colon \CT \to \CU$ and $I\colon \CS \to \CU$ be functors and further assume that there is
    \begin{itemize}
        \item a mapping on objects $L\colon \obj(\CS) \to \obj(\CT)$ and
        \item a morphism $\eta_a\colon I(a) \to RL(a)$, for every $a\in \CS$, such that for any $f\colon I(a) \to R(b)$ there is a unique ${\ol f\colon L(a) \to b}$ such that the following diagram commutes.
        \[
        \begin{tikzcd}
            I(a) \ar[swap]{rd}{f}\rar{\eta_a} & RL(a) \dar{R(\,\ol f\,)} & & L(a) \dar[dashed]{\ol f}\\
            & R(b) & & b
        \end{tikzcd}
        \]
    \end{itemize}
    Then, the mapping $L$ extends to a functor $L\colon \CS \to \CT$ and there is a natural bijection
        \begin{equation*}
            \label{eq:nat-bij}
            \rho : \hom(L(A),B) \xrightarrow{\ee\cong} \hom(I(A), R(B)).
        \end{equation*}
    witnessing that $R$ is a relative right adjoint.
\end{lem}

With this we can now prove the promised statement.

\begin{restatable}{prop}{ParametricRelativeAdjoint}
    \label{p:prr-adj}
    Assuming \ref{ax:s1}, $\lop$ satisfies \ref{ax:s2} if and only if $\lop$ is a parametric relative right adjoint.
\end{restatable}
\begin{prf}
    $(\Rightarrow)$ Assuming \ref{ax:s2} and fixing $C\in \EM \C$, we define the relative left adjoint $L_C\colon \Pa_\D \emcm \lop(C) \to \Pa_\C \emcm C$ of~$\lop_C$. Following Lemma~\ref{l:rel-adj-univ}, it is enough to define the appropriate mapping on objects
    \[ L_C\colon \obj(\Pa_\D \emcm \lop(C)) \to \obj(\Pa_\C \emcm C)\]
    and a collection of morphisms $\eta_e\colon I(e) \to \lop_C L_C(e)$, for every object $e$ of $\Pa_\C \emcm \lop(C)$ i.e. an embedding $e\colon \Pc \embed \lop(C)$ in $\EM \C$.

    Given $e\colon \Pc \embed \lop(C)$, we denote by
    \[ \Pc \xrightarrow{\ee{g_e}} \lop(\Pc_e) \xrightarrow{\lop(h_e)} \lop(C) \]
    its minimal decomposition. Define $L_C(e)$ as $h_e$ and $\eta_e$ as $g_e$. The fact that $g_e$ with $\lop(h_e)$ is a decomposition of $e$ expresses the fact that $\eta_e$ is indeed a morphism $I(e) \to \lop_C(L_C(e))$ in $\EM \D \emcm \lop(C)$, that is, we have the following.
    \begin{equation}
    \begin{tikzcd}
        \Pc \ar{rr}{\eta_e}\ar[swap]{dr}{e} & & \lop(\Pc_e) \ar{dl}{\lop(L_C(e))} \\
        & \lop(C)
    \end{tikzcd}
    \label{eq:min-decomp-f}
    \end{equation}

    Next, we check the universal property of $\eta_e$. Any morphism $z\colon I(e) \to \lop_C(g)$ in~\mbox{$\Pa_\C \emcm \lop(C)$}, by definition, corresponds to having a commutative triangle in $\EM \C$
    \[
    \begin{tikzcd}
        \Pc \ar{rr}{z}\ar[swap]{dr}{e} & & \lop(\Qc) \ar{dl}{\lop(g)} \\
        & \lop(C)
    \end{tikzcd}
    \]
    for a path embedding $g\colon \Qc \embed C$. Consequently, we have a commutative square:
    \[
    \begin{tikzcd}
        \Pc \dar[swap]{z}\rar{g_e} & \lop(\Pc_e) \dar{\lop(h_e)} \\
        \lop(\Qc) \rar{\lop(g)} & \lop(C)
    \end{tikzcd}
    \]
    By minimality of the decomposition \eqref{eq:min-decomp-f}, there is a morphism $z^*\colon \Pc_e\to \Qc$ such that
    \[
    \begin{tikzcd}
        \Pc_e \ar{rr}{z^*}\ar[swap]{dr}{h_e} & & \Qc \ar{dl}{g} \\
        & C
    \end{tikzcd}
    \]
    i.e.\ $z^*$ is a morphism $\Pc_e \to \Qc$ in $\Pa_\C \emcm C$. Furthermore, by (S1), $\lop$ preserves embeddings. Hence, $\lop(g)$ is a monomorphism and the previous two diagrams give us that $z = \lop(z^*) \circ g_e$. In other words, the following diagram commutes in~$\EM \D \emcm \lop(C)$.
    \[
        \begin{tikzcd}[column sep=1.3em]
        I(e) \ar{rr}{\eta_e}\ar[swap]{dr}{z} & & \lop_C(L_C(e)) \ar{dl}{\lop_C(z^*)} \\
        & \lop_C(g)
    \end{tikzcd}
    \]
    Unicity of $z^*$ follows from the fact that $g$ is a monomorphism.

    \medskip
    $(\Leftarrow)$ Assume $L_C$ is the relative left adjoint of $\lop_C$ from $\Pa_\C \emcm C$ to the embedding functor $I\colon \Pa_\D \emcm \lop(C) \hookrightarrow \EM \D \emcm \lop(C)$, for a fixed $C\in \EM \C$, and let $\tau$ be the corresponding natural bijection
    \[
    \tau\colon \hom(L_C(e),f) \xrightarrow{\ee\cong} \hom(I(e), \lop_C(f)).
    \]

    Next, given a path embedding $e\colon \Pc \embed \lop(C)$, observe that it is an object of~\mbox{$\Pa_\C \emcm \lop(C)$}. Set
    \[ e'\colon \Pc' \embed C  \qtq{and} e_0\colon \Pc \to \lop(\Pc')\]
    to be $L_C(e)$ and the underlying morphism of the unit $\eta_e\colon I(e) \to \lop_C L_C(e)$ of the relative adjunction, respectively, i.e.\ $e_0 = \tau(\id\colon L_C(e) \to L_C(e))$. Observe that, the fact that $\eta_e$ is a morphism in $\Pa_\C \emcm \lop(C)$ says precisely that $e$ is equal to the composition
    \[ \Pc \xrightarrow{\ee{e_0}} \lop(\Pc') \xrightarrow{\lop(L_C(e))} \lop(C). \]
    For minimality of this decomposition, consider an alternative decomposition of~$e$
    \[ \Pc \xrightarrow{\ee{g_0}} \lop(\Qc) \xrightarrow{\lop(g_1)} \lop(C) \]
    where $g_1$ is a path embedding $g_1\colon \Qc \embed C$ in $\EM \C$. The fact that it is a decomposition of $e$ expresses the fact that $g_0$ is a morphism $I(e) \to \lop_C(g_1)$ in~$\EM \D \emcm \lop(C)$. Therefore, the morphism $\tau^{-1}(g_0)$ gives $L_C(e) \to g_1$ in $\Pa_\C \emcm C$. In other words, the underlying morphism of coalgebras of $\tau^{-1}(g_0)$ is of type $\Pc' \to \Qc$ and $e'$ is equal to $g_1 \circ \tau^{-1}(g_0)$ in~$\EM \C$.
\end{prf}

\begin{rem}
    For a functor \(\lop\) to be a parametric relative right adjoint means, by Lemma~\ref{l:rel-adj-univ} specialised to \(\lop_C\), that every path embedding \(P \emb \lop(C)\) has a decomposition as \mbox{\(e_0\colon P \to \lop(P')\)} followed by \(\lop(e_1)\) for some \(e_1\colon P' \emb A\) which is minimal in the following stronger sense. For any other such decomposition \(\lop(g_1)\circ g_0\) of \(e\), not only there is a unique \(h\) such that \(e_1 = g_1 \circ h\), as in \ref{ax:s2}, but moreover we have that \(g_0 = \lop(h) \circ e_0\).
    Consequently, Proposition~\ref{p:prr-adj} shows that minimal decompositions given by \ref{ax:s2} enjoy this stronger property.
\end{rem}

\bibliographystyle{alphaurl}
\bibliography{fvmc}

\end{document}